\documentclass[letterpaper,conference]{IEEEtran}
\IEEEoverridecommandlockouts 
\usepackage{amsmath,amssymb,amsfonts,amsthm}
\usepackage{algorithmic}
\usepackage{algorithm}
\usepackage{array}
\usepackage{textcomp}
\usepackage{stfloats}

\usepackage{verbatim}
\usepackage{graphicx}
\usepackage[caption=false,font=footnotesize,
labelfont=sf,textfont=sf]{subfig}
\usepackage{cite}
\usepackage{xcolor}
\usepackage{braket}
\usepackage[nolinks=true,bookmarks=false]{hyperref}
\usepackage[T1]{fontenc}
\usepackage{gensymb}
\usepackage{orcidlink}
\usepackage{braket}
\usepackage{subfig}

\usepackage{mathtools}
\mathtoolsset{centercolon}

\theoremstyle{plain}
\newtheorem{thm}{Theorem}

\newtheorem{lem}{Lemma}
\newtheorem{cor}{Corollary}

\newtheorem{rk}{Remark}
\newtheorem*{attr*}{Attribution}

\theoremstyle{definition}

\def\BibTeX{{\rm B\kern-.05em{\sc i\kern-.025em b}\kern-.08em
    T\kern-.1667em\lower.7ex\hbox{E}\kern-.125emX}}

\DeclareMathOperator{\erfc}{erfc}
\newcommand{\dif}{\mathop{}\!\mathrm{d}}

\allowdisplaybreaks

\begin{document}

\title{Toward Practical Two-Way Covert Communication 
}

\author{\IEEEauthorblockN{   Paul N.~Fessatidis\IEEEauthorrefmark{1}, Wyatt Wallis\IEEEauthorrefmark{1}\IEEEauthorrefmark{2}, Tae E. Cooper\IEEEauthorrefmark{1}, Mark J. Meisner\IEEEauthorrefmark{4}, Jaim Bucay\IEEEauthorrefmark{4},\\
Saikat Guha\IEEEauthorrefmark{5}, Shelbi L. Jenkins\IEEEauthorrefmark{1}\IEEEauthorrefmark{2},  Michael S. Bullock\IEEEauthorrefmark{3}\IEEEauthorrefmark{1}, and Boulat A. Bash\IEEEauthorrefmark{1}\IEEEauthorrefmark{2}} \\
		
    \IEEEauthorblockA{\normalsize
        \IEEEauthorrefmark{1} Electrical and Computer Engineering Department, University of Arizona, Tucson, AZ, USA \\
        \IEEEauthorrefmark{2} Physics Department, University of Arizona, Tucson, AZ, USA  \\
        \IEEEauthorrefmark{3} Manning College of Information and Computer Sciences, University of Massachusetts, Amherst, MA, USA\\
        \IEEEauthorrefmark{4} Raytheon, Tucson, AZ, USA\\
        \IEEEauthorrefmark{5} Department of Electrical and Computer Engineering, University of Maryland, College Park, MD, USA
        }\thanks{Paul N.~Fessatidis and Wyatt Wallis are co-first authors. This work was supported by Raytheon. This document does not contain technology or technical data controlled under either U.S.~International Traffic in
Arms Regulation or U.S.~Export Administration Regulations.
The views expressed here are those of the authors and should not be interpreted as representing any employer
or funding organization.
}}

\maketitle

\begin{abstract}

We study two-way covert communication schemes, where information is transmitted by passively modulating a reflected signal back to the source. We consider optical systems, described by quantum bosonic channels.
While broadband classical and quantum light sources offer high covert throughput in theory, the associated mode-matching  and phase-synchronization requirements make them impractical. 
Therefore, we employ a narrowband laser source to experimentally demonstrate a proof-of-concept two-way covert communication system, where the adversary is assumed to be quantum-capable.
Furthermore, we propose a correlator-based receiver that attains the broadband gain offered by a quantum light source without the need for precise mode matching.   
\end{abstract}

\section{Introduction}
\label{sec:intro}
Covert, or low probability of detection/intercept (LPD/LPI) communication, aims to reliably transmit information while avoiding detection by an adversary \cite{bash12sqrtlawisit, bash13squarerootjsacnonote, bloch15covert, wang15covert,bash15covertcommmag, chen23covcommssurvey}. This contrasts standard encryption methods, which merely prevent an adversary from accessing the contents of a transmission \cite{menezes96HAC}. The stricter requirement of covert communication has been well studied for classical channels over the past decade, yielding the \emph{square root law} (SRL), which states that the number of covert bits transmitted scales as $\sqrt{n}$ over $n$ channel uses\cite{bash12sqrtlawisit, bash13squarerootjsacnonote, bloch15covert, wang15covert,bash15covertcommmag, chen23covcommssurvey}. While the Shannon capacity of the covert channel is zero, since  $\lim_{n\to\infty}\frac{\sqrt{n}}{n}=0$, the SRL still allows transmission of many covert bits for large $n$. 

Analysis of the fundamental limits of optical communication systems requires a quantum mechanical description of the system \cite{wilde16quantumit2ed}, which led to the development of the SRL over bosonic channels\cite{bash15covertbosoniccomm, bullock20discretemod, gagatsos20codingcovcomm}. 
However, experimental exploration of SRL-based covert communication remains scarce, with only a few existing works \cite{bash15covertbosoniccomm, liu24metrofibercovcomm, Djordjevic25covertSLMatmospheric, Djordjevic25covertSLM, bali25covertsdr-milcom, bali26covertsdr}. 
Our work is inspired by the recent results in covert phase sensing \cite{bash19covertsensor, gagatsos19floodlightsensorcleo, gagatsos19floodlightsensor, Hao_2022}, which show that broadband sources yield substantial performance gains under strict power constraints imposed by covertness requirements.
In this scenario, the mean square error of a phase estimator is governed by the SRL with respect to the number of optical modes of the probe \cite{bash17qcovertsensingisit}, and is substantially improved when using a broadband thermal state as a probe and as a local resource for a heterodyne detection measurement \cite{gagatsos19floodlightsensorcleo, gagatsos19floodlightsensor}.

We transform these covert sensing techniques into a two-way covert communication scheme, where Alice transmits displaced thermal states to Bob, which he phase modulates to encode information and returns to Alice. Adversary Warden Willie tries to determine if communication between the two parties occurs. We find that covert communication benefits from the use of thermal sources due to their vastly greater optical bandwidth. However, using broadband sources requires precise mode matching at the receiver. The performance sacrificed using laser inputs is offset by their viability in experimental demonstrations. 
Thus, we employ narrowband lasers for our proof-of-concept demonstration of two-way covert communication.
Furthermore, we propose a new transceiver that uses a quantum two-mode squeezed vacuum (TMSV) light source and a correlator. It is based on recent clock synchronization schemes \cite{Lee_2022} and provides broadband gain without the mode-matching cost. 

Our paper is organized as follows.
In Section \ref{sec:prerequisite}, we introduce the two-way covert communication model and formulate the covertness criterion. In Section \ref{2-wayconsid}, we discuss the strict input power constraints imposed by covertness and adapt the sparse signaling scheme from \cite{tahmasbi21signalingcovert, tan26covertsignaling-arxiv} to a practical scenario where the modulator bandwidth is smaller than the optical bandwidth of the source. We then discuss the benefits and drawbacks of broadband and narrowband input sources and present theoretical analysis of covertness requirement that generalizes to both. In Section \ref{sec:Experiment_Setup}, we present experimental results demonstrating the proof-of-concept two-way covert communication system using a laser source, before proposing a novel transceiver structure in Section \ref{sec:CCR} that employs TMSV and a correlator. We compare its covert throughput with a several phase-sensitive schemes, showing its superior performance. We conclude in Section \ref{sec:discussion}. 

 \section{Prerequisites}
 \label{sec:prerequisite}
 \subsection{Bosonic channel}
Our fundamental transmission unit is a single orthogonal mode of the electromagnetic field at some wavelength (e.g., optical).
A mode is akin to a channel use in information theory.
We assume $n=TW$ available modes, where $TW$ is the time-bandwidth product.
Lossy thermal-noise bosonic channel models quantum-mechanically describe the transmission of a single mode over linear loss $\eta$ and additive Gaussian noise (such as noise stemming from blackbody radiation) with mean photon number $\bar{n}$.

It is described by a beamsplitter with transmittance $\eta$ where the input from the environment is a thermal state with mean photon number $\bar{n}\geq 0$, $\hat{\rho}_{th}(\bar{n}) = \sum_{k=1}^\infty \frac{\bar{n}^k}{(\bar{n}-1)^{k+1}}\ket{k}\bra{k}=\int_\mathbf{} \dif ^2\beta \frac{\exp[\frac{-|\beta|^2}{\bar{n}}]}{\pi \bar{n}} \ket{\beta}\bra{\beta}$, represented in the photon-number and coherent-state bases, respectively. Photon-number states, $\{ \ket{k}:k\in\mathbb{Z}_{\geq0}\}$ are defined by the number of photons in the given mode, $k$, and are the eigenstates of the photon-number operator, $\hat{N}\ket{k}=\hat{a}^\dagger\hat{a}\ket{k}=k\ket{k}$, where $\hat{a}$ and $\hat{a}^\dagger$ are the annihilation and creation operators for the given mode such that $\hat{a}\ket{k}=\sqrt{k}\ket{k-1}$ and $\hat{a}^\dagger\ket{k}=\sqrt{k+1}\ket{k+1}$. Coherent states are eigenstates of the annihilation operator, $\hat{a}\ket{\beta}=\beta\ket{\beta},\;\beta\in\mathbb{C}$, and represent ideal laser sources. Note that $\bar{n}= 0$ for the lossy thermal-noise bosonic channel reduces to the pure loss channel without injected noise. 
 
\subsection{Two-Way Channel Model}
\label{subsec:2way}
Our two-way channel model, depicted in Fig.~\ref{fig:channel_model}, contains two lossy thermal-noise bosonic channels, each parameterized by transmittance $\eta_i$ and thermal-noise mean photon number $\bar{n}_{B_i}$, $i\in\{1,2\}$. It is inspired by the theory \cite{bash19covertsensor, gagatsos19floodlightsensorcleo, gagatsos19floodlightsensor} and experiments in \cite{Hao_2022}. These works show that covert sensors based on broadband sources such as amplified spontaneous emission (ASE) and spontaneous parametric down-conversion (SPDC), significantly outperform the narrowband laser ones. 
This is because the former provide vastly greater number of available modes for the same integration time.
\begin{figure}[htpb]
    \centering
    \includegraphics[width=0.85\linewidth]{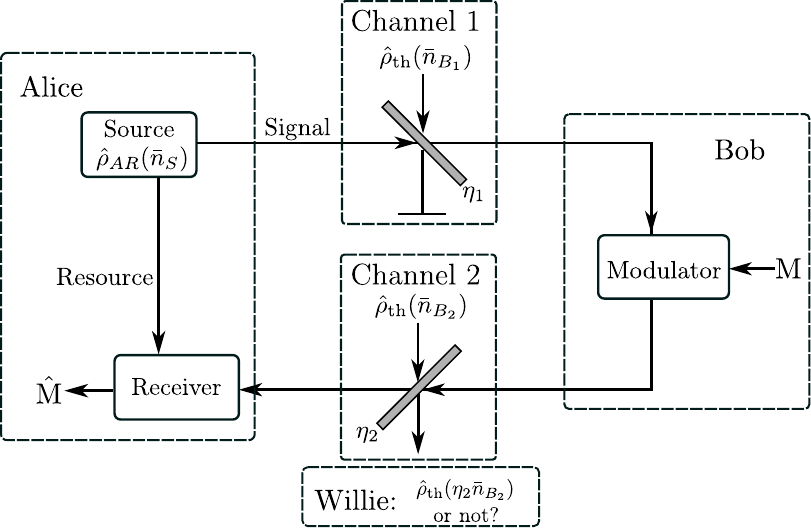}
    \caption[Two-way covert communication model where Willie stands behind Alice.]{Two-way covert communication model where Willie stands behind Alice, described in Section \protect\ref{subsec:2way}.
    }
\label{fig:channel_model}

\end{figure}
Alice can generate signals, retain local resources in quantum memories, transmit to Bob over channel 1 and receive from Bob over channel 2. Bob can receive from Alice over channel 1, modulate received signals to encode information, and transmit to Alice over channel 2. We restrict Bob to \emph{passive} operations that cannot amplify or otherwise add energy into channel 2 not provided by channel 1. We assume Willie has access to arbitrary quantum measurement, memory, computing resources, and complete knowledge of transmission parameters, allowing him to filter all but the relevant band used for communication. However, 
he can only access photons lost to the environment in either channel 1 or channel 2. That is, he can observe either Alice or Bob, but not both.
This corresponds to restricting him outside the line of sight in free-space optical communication.
Now, if Willie chooses to observe Alice via channel 1 environment output, Bob can modulate thermal noise from channel 1 and transmit it on channel 2 undetected (since this does affect the channel 1 environment).
Hence, Willie observes Bob via the environment output for channel 2.
Bob's no-transmission ``innocent'' input into channel 2 is vacuum, i.e., he reflects no light from either Alice or the channel 1 environment. 

\subsection{Hypothesis Testing and Covertness Criterion}
\label{sec:covertnesscriterion}

Given access to channel 2 environment, Willie's goal is to detect whether transmission has occurred. 
When Bob does not transmit, Willie receives an $n$-fold thermal state $\hat{\rho}_0^{W^n}=\hat{\rho}_{\eta\bar{n}}^{\otimes n}$, otherwise he receives another state $\hat{\rho}_1^{W^n}$.
Assuming equally likely prior probabilities on whether or not Bob transmits, Willie's decision error probability is $P_e^{(w)} = \frac{P_{\text{FA}}+P_{\text{MD}}}{2}$, where $P_{\text{FA}}$ and $P_{\text{MD}}$ are the probabilities of false alarm and missed detection. The quantum-optimal receiver yields $P_{e,{\text{min}}}^{(w)}=\frac{1}{2}-\frac{1}{4}\|\hat{\rho}_1^{W^n}-\hat{\rho}_0^{W^n}\|_1$, where $\|\hat{\rho}-\hat{\sigma}\|_1$ is the trace distance between  quantum states $\hat{\rho}$ and $\hat{sigma}$ \cite{wilde16quantumit2ed}. Therefore, we say \emph{a system is covert if, for some} $\delta_p>0$, $P_e^{(w)}\geq\frac{1}{2}-\delta_p$. However, the trace distance is mathematically cumbersome, while the quantum relative entropy (QRE) is more convenient and relates to the trace distance via Pinsker's inequality, $\|\hat{\rho}-\hat{\sigma}\|_1 \leq \sqrt{2D(\hat{\rho}\|\hat{\sigma})}$. Thus, we use the following covertness criterion: \emph{a system is covert if, for some} $\delta_{\text{QRE}}>0$, $D(\hat{\rho}_1^{W^n}\|\hat{\rho}_0^{W^n})\leq \delta_{\text{QRE}}=2\delta_p^2$.

\section{Practical Considerations for Two-way Covert Communication}
\label{2-wayconsid}

\subsection{Signaling Scheme}
The covertness criterion in Section \ref{sec:covertnesscriterion} imposes a strict limit on Bob's output power: he can transmit $\propto\sqrt{\frac{\delta_{\text{QRE}}}{n}}$ photons per mode on average over $n$ modes \cite{bullock20discretemod}.
This requires the use of either \emph{diffuse} or \emph{sparse} signaling \cite{tahmasbi21signalingcovert, tan26covertsignaling-arxiv}. The former spreads the total $\propto\sqrt{n}$-photon energy budget over all $n$ available modes whereas the latter uses a constant photon number in $\propto\sqrt{n}$ out of $n$ randomly selected modes. Although diffuse signaling can improve covert performance (e.g., using entanglement assistance \cite{gagatsos20codingcovcomm}), it is impractical due to finite resolution of digital-to-analog and analog-to-digital converters (DAC/ADC). Thus, sparse signaling is used in covert communication experiments, including \cite{bash15covertbosoniccomm, bali25covertsdr-milcom, bali26covertsdr}. 

However, direct application of sparse signaling is challenging for broadband sources, as transmitting on individual modes requires a high-bandwidth intensity modulator. Thus, we consider \emph{symbol-sparse signaling}, where $n$ available modes are partitioned into $n/M$ symbol slots of $M$ modes each.
Bandwidths $W$ and $W_M$ of Alice's transmission and Bob's modulator, respectively, determine the symbol length $M=\frac{W}{W_M}$. 
Alice and Bob flip a biased coin $n/M$ times.
Heads have probability $\tau$ and indicate the symbol slots used for transmission; 
tails indicate the unused slots.
The resulting binary sequence is part of Alice and Bob's shared secret.

\subsection{Input Signal Modulation}
\label{subsec:input_signal}
Large number of optical modes offered by broadband sources such as ASE and SPDC provide significant performance gains for covert two-way systems.
However, those employing retained reference (e.g., optical homodyne, phase-conjugate receiver (PCR), and others \cite{guha09quantumilluminationOPA, zhuang17ff-sfg}) need it to be mode-matched to the returned signal.
This translates to requiring the reference and signal path lengths to be within the coherence length of the source, which is inversely proportional to its bandwidth. 
For example, the coherence length is $\approx 150~\mu$m for the $W=2$ THz sources in \cite{Hao_2022}.
Mobility, atmospheric turbulence, and other artifacts challenge attaining such precision in practice.
Thus, we consider a phase insensitive receiver for broadband SPDC source in Section \ref{sec:CCR}, and a narrow-band coherent laser source. 
The latter trades theoretical performance for practicality and enables the two-way covert communication experiment in Section \ref{sec:Experiment_Setup}. 

Suppose Alice uses a quasi-monochromatic laser source followed by an erbium doped fiber amplifier (EDFA).
Her input symbol is then described by $\hat{\rho}_{\text{th}}(\bar{n}_S,\alpha_0)\otimes\hat{\rho}_{\text{th}}(\bar{n}_S)^{\otimes M-1}$, where  $\hat{\rho}_\text{th}(\bar{n}_S,\alpha_0)=\int_\mathbf{} \dif ^2\beta \frac{\exp[\frac{-|\beta-\alpha_0|^2}{\bar{n}_S}]}{\pi \bar{n}_S} \ket{\beta}\bra{\beta}$ is a displaced thermal state with displacement $\alpha_0=|\alpha_0|\exp(j\theta_0) \in \mathbb{C}$, phase $\theta_0$, and mean thermal photon number $\bar{n}_S$. Note that the displaced thermal state occupies \emph{a single mode}, and all other modes are thermal. Channel 1 evolves the state to $\hat{\rho}_{\text{th}}(\bar{n}_S^\prime,\sqrt{\eta_1}\alpha_0)\otimes\hat{\rho}_{\text{th}}(\bar{n}_S^\prime)^{\otimes M-1}$ where $\bar{n}_S^\prime= (1-\eta_1)\bar{n}_{B_1} + \eta_1\bar{n}_{S}$. 
The total mean photon number per mode of this state, $\bar{n}_T=\bar{n}_\alpha/M+\bar{n}_{S}^\prime$ with $\bar{n}_\alpha=\eta_1|\alpha_0|^2$, is constrained by the covertness criterion from Section \ref{sec:covertnesscriterion}.

For phase $\theta$, Willie's input state has the same form, having evolved through channels 1 and 2: $\hat{\rho}_{\text{th}}(\bar{n}_W,\alpha_\theta)\otimes\hat{\rho}_{\text{th}}(\bar{n}_W)^{\otimes M-1}$ where $\alpha_\theta=\sqrt{\eta_1(1-\eta_2)}\alpha_0\exp(j\theta)$ and $\bar{n}_W= (1-\eta_2)\bar{n}_S^\prime + \eta_2\bar{n}_{B_2}$.
Recall that, when Bob does not transmit, Willie's input is $\hat{\rho}_0^{W^n}=\hat{\rho}_{\text{th}}(\eta_2\bar{n}_{B_2})^{\otimes n}$. Although Willie is aware of the signaling scheme, his ignorance of Alice and Bob's pre-shared secret results in the mixed input state when Bob transmits a symbol with phase $\theta$: 
$\hat{\rho}_1^{W^M}(\theta)=(1-\tau)\hat{\rho}_{\text{th}}(\eta_2\bar{n}_{B_2})^{\otimes M}+\tau\hat{\rho}_{\text{th}}(\bar{n}_W,\alpha_\theta)\otimes\hat{\rho}_{\text{th}}(\bar{n}_W)^{\otimes M-1}$.
Alice and Bob use quadrature phase-shift keying (QPSK) modulation, and deny Willie the symbol-phase information via a pre-shared secret. Therefore, Willie's overall input state is $\hat{\rho}_1^{W^n}=\left(\frac{1}{4}\sum_{q=0}^3\hat{\rho}_1^{W^M}(q\pi/2)\right)^{\otimes n/M}$.
This allows characterizing the symbol transmission probability, $\tau$, for the symbol-sparse signaling scheme. 

\begin{thm}\label{thm:sparse_signaling}
Suppose Alice and Bob employ symbol-sparse signaling for QPSK modulated displaced thermal states. For a fixed return signal mean photon number per mode, $\bar{n}_T\geq0$, $\tau\leq  \frac{1}{(1-\eta_2)} \sqrt{\frac{2\eta_2\bar{n}_{B_2}(1+\eta_2\bar{n}_{B_2})}{\bar{n}_\alpha^2+2\bar{n}_\alpha\bar{n}_S^\prime+M\bar{n}_S^{\prime2}}}\sqrt{\frac{\delta_\text{QRE}}{n/M}}$ maintains the covertness criterion, $D(\hat{\rho}_1^{W^n}\|\hat{\rho}_0^{W^n})\leq \frac{\delta_{\text{QRE}}}{n}$. 
\end{thm}

The proof is in Appendix \ref{ap:covertness}. 
The following corollary provides the fundamental limit of symbol-sparse signaling for covert communication:

\begin{cor} \label{cor:sparse_throughput}
    Suppose that the transceiver choice and channel parameters induce a classical channel between Bob and Alice with Shannon capacity $C>0$ for an $M$-mode symbol. For $n$ large enough and $\tau$ selected according to Theorem \ref{thm:sparse_signaling}, there exists a code $\mathcal{C}$ that allows transmitting at least $L\sqrt{n}$ bits with decoding error probability $\bar{P}_e(\mathcal{C})\leq \epsilon_c, \text{ and } D\left(\hat{\bar \rho}^n\middle\|\hat{\rho}_{0}^{\otimes n}\right) \leq \delta_{\rm QRE}$, where $L=(1-\vartheta) \tau C\sqrt{n/M}$ is the \emph{covert-capacity lower bound}, $\hat{\bar \rho}^n$ is Willie's output state when Bob transmits (averaged over $\mathcal{C}$), $\hat{\rho}_{0} = \hat{\rho}_{th}(\eta_2\bar{n}_{B_2})$ is Willie's output when Bob is quiet, and  $\epsilon_c,\delta_{\rm QRE}, \vartheta > 0$ are constants.
\end{cor}

The proof follows from the sparse signaling code construction and analysis in \cite[Lem.~5]{anderson2025covertentanglementgenerationbosonic}, with quantum error correction code replaced with a classical error correction code, suitable substitution of states in the covertness analysis and employing Theorem \ref{thm:sparse_signaling}. 
Note that in the proof therein the number of selected symbol slots to use for transmission can be set arbitrarily close to $\tau (n/M)$ via control of $\vartheta$. 

\section{Experiment Implementation and Results}
\label{sec:Experiment_Setup}

\subsection{Experiment design}
\label{subsec:exp_design}
Our experiment design is in Fig.~\ref{fig:DESIGN}. Alice generates a signal at 1550 nm using a continuous wave Amonics ADFB-1550-20-B laser. It is split along the signal and reference arms. The signal arm emulates transmission to Bob and is attenuated with a Thorlabs V1550A attenuator. Bob modulates a square-wave binary phase-shift keying (BPSK) signal using a Thorlabs LN65S electro-optic modulator.
Thermal noise from Pri-tel FA-18 EDFA  is filtered by a wavelength-division multiplexer (WDM) to a 30.508 nm width centered at 1550 nm and injected into the return path using a 50-50 splitter. The resulting noisy signal is output to Alice and Willie. We measure the signal and noise power at Willie's output using a Yokogawa AQ6370B optical spectrum analyzer (OSA). This allows us to determine the values of $\eta_2\bar{n}_{B_2}$ and $(1-\eta_2)\bar{n}_\alpha$, which we then use with Theorem \ref{thm:sparse_signaling} and Corollary \ref{cor:sparse_throughput} to determine covert throughput that our system allows. 

\begin{figure}[ht]
\centering
    \includegraphics[width=0.35\textwidth]{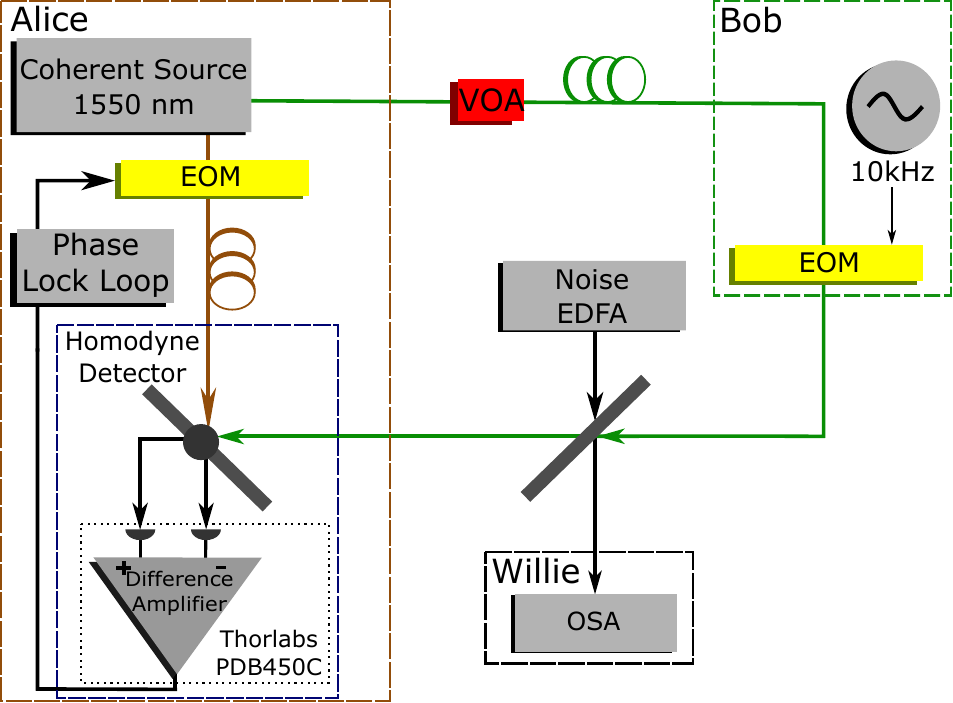}
\caption[Experiment Setup]{
Experimental setup. The coherent source (Amonics ADFB-1550-20-B) is split into the signal (green) and reference (orange) arms. Bob and Alice's EOM's are Thorlabs LN65S. A Pri-tel FA-18 EDFA injects noise into the channel. We use Yokogawa AQ6370B OSA to collect data at Willie, and a Thorlabs PDB450C balanced detector at Alice.
}
    \label{fig:DESIGN}
\end{figure}

Alice demodulates Bob's BPSK signal using optical homodyne detection. Her output is mixed with the reference arm using a variable $2\times 2$ splitter, which ensures the mixed signals are matched in amplitude before entering the ThorLabs PDB450C balanced detector. We measure its output using a LeCroy SDS120X oscilloscope. We correct phase errors due to, e.g., fiber bends and  polarization mismatches, with a phase-lock loop (PLL). We construct it using an AD835 multiplier and SRS SIM960 proportional-integral-derivative (PID) controller. We emphasize that we use PLL, as in \cite{Hao_2022}, for  convenience. In practical covert communication, the long gaps between transmissions imposed by sparse signaling preclude its use. Instead, we anticipate known pilot symbols to be transmitted with each data symbol to ensure synchronization, as in \cite{bali25covertsdr-milcom, bali26covertsdr}. We defer this to future work.

\begin{figure*}[t]
    \centering
    \subfloat[Normalized transmission probability $\tau \sqrt{n}$]{
        \includegraphics[width=0.32\textwidth]{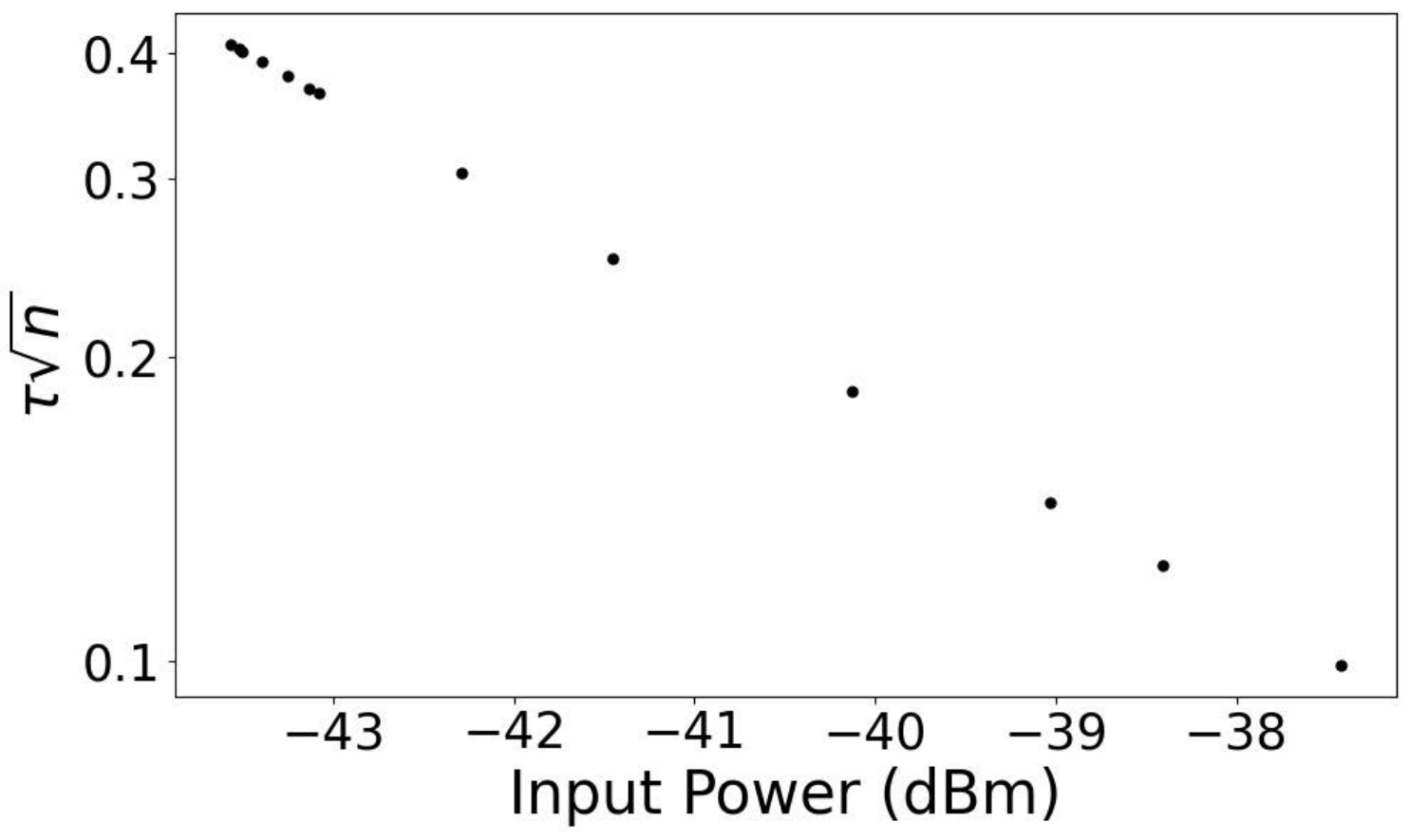}\label{fig:TAU}
        }
        \hfill
    \subfloat[BER and Shannon capacity]{
        \includegraphics[width=0.32\textwidth]{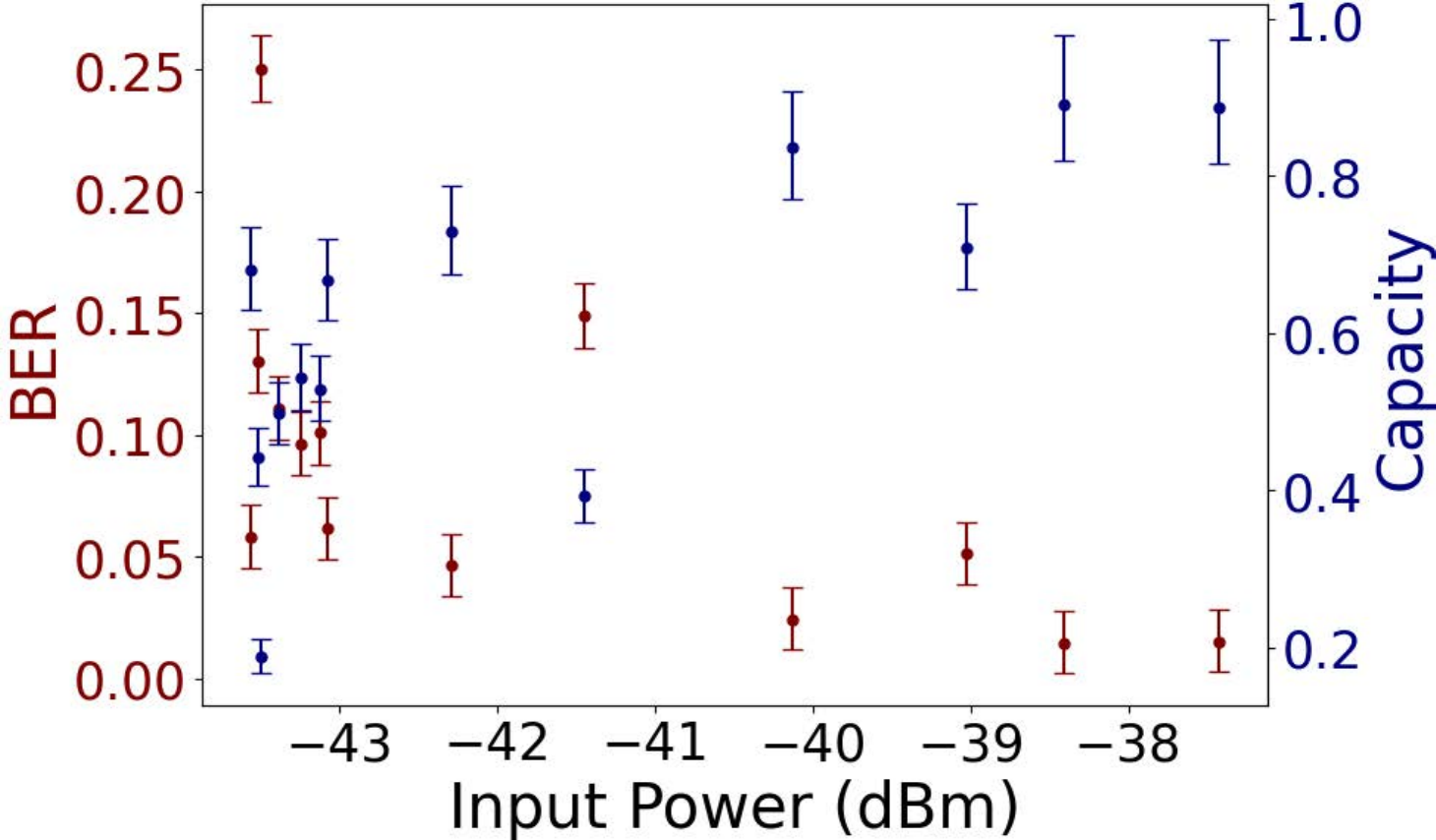} \label{fig:BER}
        }
        \hfill
    \subfloat[Covert capacity lower bound]{\includegraphics[width=0.32\textwidth]{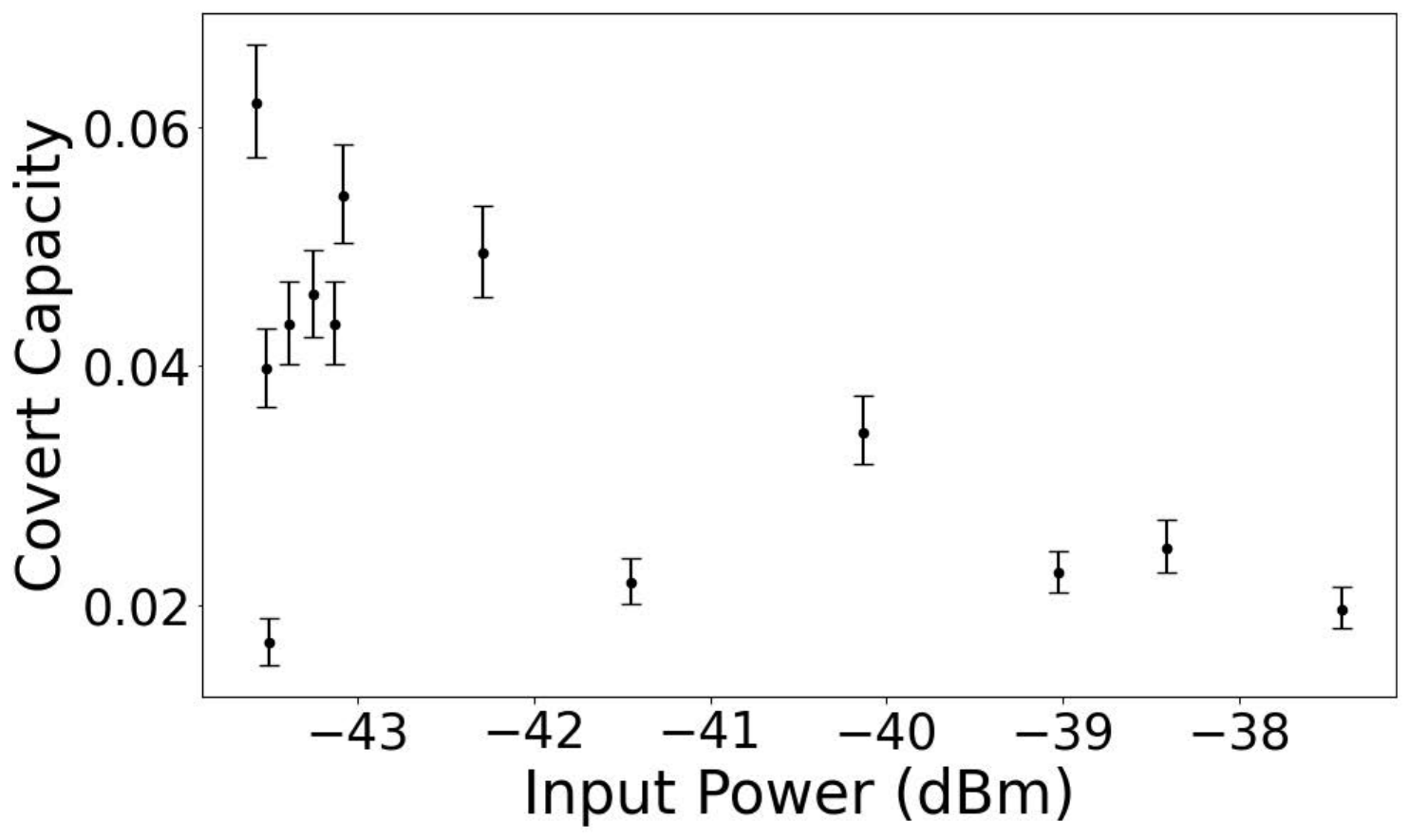}\label{fig:Cov_Cap}
    }
    \caption[Experiment results]{Experimental results discussed in Section \protect\ref{subsec:analysis}. 
    \label{fig:RESULTS}}

\vspace{-10pt}
\end{figure*}

\subsection{Data Collection and Results} 
\label{subsec:analysis}
Our results are in Fig.~\ref{fig:RESULTS}.
We obtain the lower bound on covert capacity $L$ for our design using Corollary \ref{cor:sparse_throughput} and the data from the experiment.
We set $\bar{n}_S^\prime=0$ in Theorem \ref{thm:sparse_signaling} since Alice uses a coherent state and $\eta_1=1$ since loss on the path from Alice to Bob is negligible. Thus, $L\geq  (1-\vartheta)\frac{\sqrt{2\eta_2\bar{n}_{B_2}(1+\eta_2\bar{n}_{B_2})}}{(1-\eta_2)\bar{n}_\alpha}\sqrt{\delta_\text{QRE}}C $ is a function of the mean photon numbers $(1-\eta_2)\bar{n}_\alpha$ and $\eta_2\bar{n}_{B_2}$ that Willie receives from Alice and noise source, respectively, Alice and Bob's tolerance to detection $\delta_{\text{QRE}}=0.05$, and the Shannon capacity $C$ per transmitted symbol.

We employ Theorem \ref{thm:sparse_signaling}'s expression for $\tau$, even though we use BPSK rather than QPSK. This is because, in principle, Alice can create random QPSK signal at Willie by adding and subtracting a random phase on her outgoing and incoming signals. We omit this from our setup for simplicity.
We use the OSA at Willie's output to obtain $\bar{n}_\alpha=\frac{S(\lambda)}{r} \frac{\lambda^3}{hc^2}$, where $S(\lambda)$ is the power at the signal's peak wavelength $\lambda$ from an OSA independently calibrated with a photodiode, $r =0.02$ nm is the resolution bandwidth of the OSA, $h$ is the Plank constant, and $c$ is the speed of light. The noise source is off when we collect the trace for $\bar{n}_T$.  We obtain $\bar{n}_{B_2}$ similarly by collecting the trace from the noise source  with signal off.
We plot $\tau\sqrt{n}$ in Fig.~\ref{fig:TAU}.

We use the oscilloscope at Alice's output to obtain Shannon capacity $C$ as follows: we estimate the means $\{V_0,V_1\}$ and standard deviations $\{\sigma_0,\sigma_1\}$ corresponding to BPSK zero and one pulses. We employ $T=1.288$ ms pulse interval at $10^7$ samples/s rate, and collect $n_d=12,288$ readings for each pulse type.
From these, we estimate Alice's bit error rate $P_e^{(a)}=\frac{1}{2}\erfc\left(\frac{V_1-V_0}{\sqrt{2}(\sigma_1+\sigma_0)}\right)$, where $\erfc(x)=\frac{2}{\sqrt{\pi}}\int_x^\infty e^{-t^2}\dif t$ is the complementary error function \cite[Ch.~4.5]{2007fiber}.
Wilson score yields the $\alpha=95$\% confidence interval $\left(P_e^{(a),-},P_e^{(a),+}\right)$ for $P_e^{(a)}$ \cite[Eq. 8.27]{app_stat}:
        $P_e^{(a),\pm} =\frac{1}{1+z_\alpha^2/n_d}\left(P_e^{(a)}+\frac{z_\alpha^2}{2n_d} \pm \frac{z_\alpha}{2n_d} \sqrt{4n_dP_e^{(a)}(1-P_e^{(a)}+z_\alpha^2)}\right)$,
where $z_\alpha=1.96$ is the critical value.
Then, $C=H_2\left(P_e^{(a)}\right)$, where $H_2(p)=-p\log_2 p-(1-p)\log (1-p)$ is the binary entropy function \cite[Ch.~7.1.4]{cover02IT}, with confidence intervals computed from $\left(P_e^{(a),-},P_e^{(a),+}\right)$.
Fig.~\ref{fig:BER} plots $P_e^{(a)}$ and $C$.

Letting $\vartheta=0$ and multiplying $\tau\sqrt{n}$ from Fig.~\ref{fig:TAU}, $C$ from Fig.~\ref{fig:BER}, and $\sqrt{\delta_{\text{QRE}}}$ yields the $L$ plotted in Fig.~\ref{fig:Cov_Cap}. The input power yielding its maximum is between $-42.33$ and $-40.57\; \text{dBm}$. Thus, increasing input power does not benefit the covert capacity. 

\section{Proposed Correlator-based Receiver}
\label{sec:CCR}
\begin{figure}[htpb]
\centering
    \includegraphics[width=0.45\textwidth]{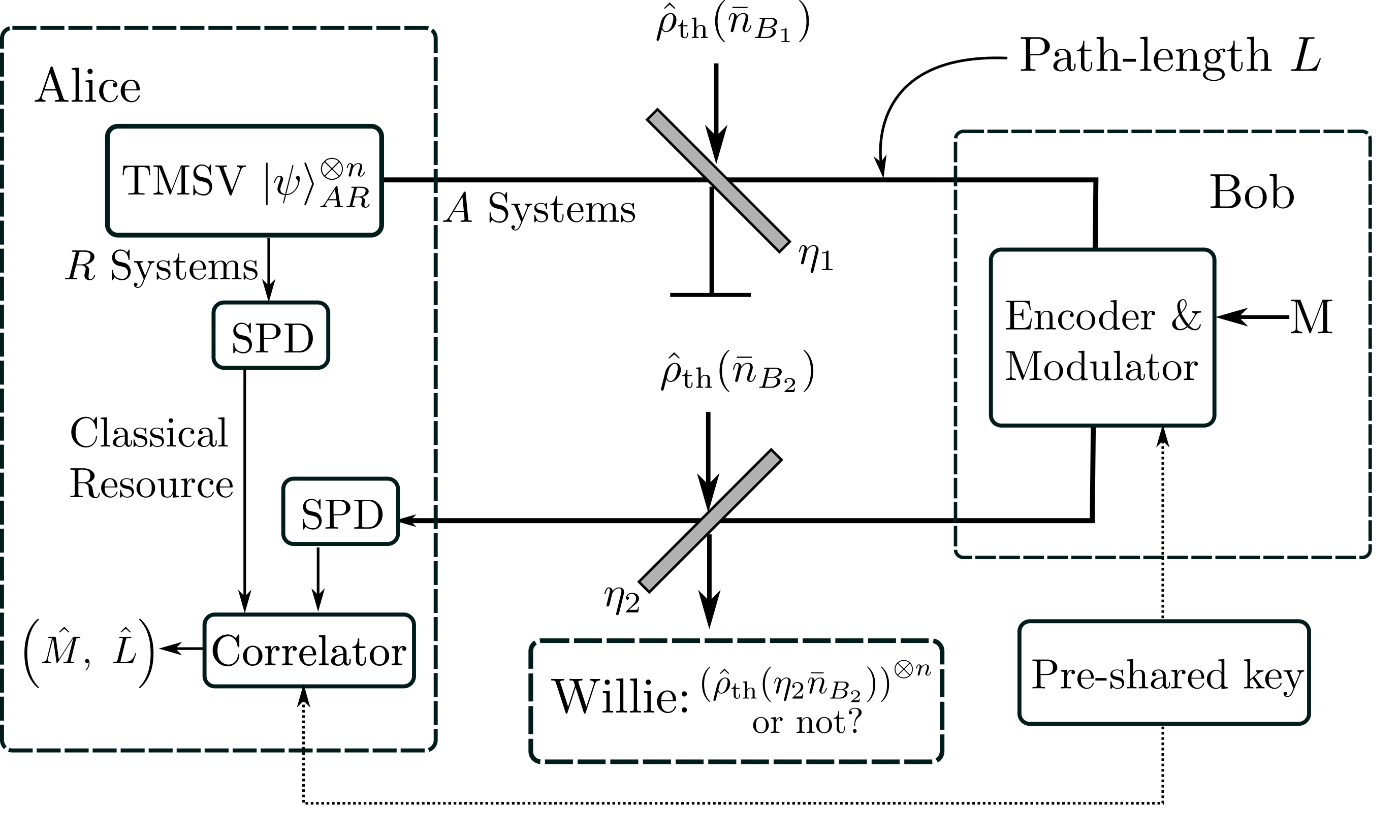}
\caption[Configuration for the proposed correlator-based covert-communication and range-sensing transceiver.]{Configuration for the proposed correlator-based covert-communication and range-sensing transceiver.}
    \label{fig:tmsv_spd_diagram}
    \vspace{-5pt}
\end{figure}

Now we propose a phase-insensitive receiver, depicted in Fig.~\ref{fig:tmsv_spd_diagram}, as a practical alternative to the broadband schemes based on \cite{bash19covertsensor, gagatsos19floodlightsensorcleo, gagatsos19floodlightsensor, Hao_2022}.
Suppose that Bob and Alice have knowledge of the upper and lower bound $l_l$ and $l_u$ of an unknown path length $l$. 

 Alice prepares a product of $n$ two-mode squeezed vacuum (TMSV) states, $\ket{\psi}_{AR}^{\otimes n}=\sum_{k=1}^\infty \sqrt{\frac{\bar{n}^k}{(\bar{n}-1)^{k+1}}}\ket{k}_A\ket{k}_R$ which describes a maximally-entangled state of signal and reference systems $A$ and $R$. The $n$ reference systems are immediately measured by Alice using a single photon detector (SPD), generating an $n$-length i.i.d.~binary sequence, $\mathbf{D}^{(i)}$. The $n$ signal systems are transmitted over channel 1 to Bob. Using the preshared secret, Bob uses on-off keying (OOK) modulation with probability of transmitting an on-symbol $p_{\text{OOK}}$ on each of $N_B = (1-\vartheta)\tau (n/M) = \lfloor \tau (n/M)\rfloor$ slots, where $\tau$ is from Theorem \ref{thm:sparse_signaling}, and discards the rest. 
 
 Suppose that Alice knows $l$. She turns the second SPD on when she expects the light to arrive, generating a second $n$-long binary sequence, $\mathbf{D}^{(s)}$. Detection events from both sequences corresponding to the unused symbol slots are discarded, yielding the subselected sequences, $\mathbf{D}^{*(i)}$ and $\mathbf{D}^{*(s)}$. Alice computes the classical correlation $C_u = \sum_{k=1}^M \mathbf{D}^{*(i)}_{k,u}\mathbf{D}^{*(s)}_{k,u}$ on each symbol slot $u \in 1,...,N_B$. Note that $C_u$ is a binomial distributed random variable with $M$ trials and a probability of success that depends on $p_{\text{OOK}}$. Alice demodulate each symbol using $C_u$. This results in a binary asymmetric channel, characterized in the Appendix \ref{ap:ccr_performance}. 

If Alice does not know $l$, she uses $\mathbf{D}^{(i)}$ and $\mathbf{D}^{(s)}$ to determine it without altering the overall transmission strategy, i.e., \emph{without impacting covertness or information throughput}. She subdivides the difference between $l_l$ and $l_u$ into $N_R$ discrete bins of the mode-length, $c/W$, where $c$ is the speed of light and $W$ is  the transmission bandwidth. For $v\in \{i,s\}$, Alice takes the subselected sequence, $\mathbf{D}^{*(v)}$, and pads each symbol by the $M+N_R$ elements of $\mathbf{D}^{(v)}$ surrounding the symbol via $l_l$ and $l_u$, yielding an $N=(N_R+2M)N_B$ -long sequence, $\mathbf{D}_p^{(v)}$. By taking the argmax of the cross-correlation between $\mathbf{D}_p^{(i)}$ and $\mathbf{D}_p^{(s)}$, Alice can recover $l$ within $c/W$ based on the chosen index. Appendix \ref{ap:ccr_performance} proves that this results in the correct $l$ with certainty as $n \rightarrow\infty$. 
Note that this method has been experimentally verified in the context of clock synchronization over pure-loss fiber channels \cite{Lee_2022}. 

We compare covert throughput of the following schemes:
\begin{enumerate}
    \item Our correlator-based proposal discussed above.
    \item TMSV with PCR experimentally demonstrated in \cite{Hao_2022}, assuming a receiver gain coefficient of $1+0.257\times10^{-3}$.
    \item Thermal with homodyne receiver also demonstrated in \cite{Hao_2022}, assuming a resource power of $6\times10^3$ photons per mode ($\sim1$ mW for the broadband signal).
    \item Coherent OOK with PNR, experimentally demonstrated with SPDs in \cite{bash15covertbosoniccomm}.
    \item Coherent BPSK with homodyne receiver in Section \ref{sec:Experiment_Setup}.
\end{enumerate}
Schemes $\{1,4\}$ and $\{2, 3, 5\}$ modulate, respectively, intensity and phase. Thus, our channel 1 is lossless to not favor the former.
We compute the induced binary channel capacities for the following model.
We use the MODTRAN ``Mid-Latitude Summer (MLS)'' atmospheric profile \cite{berk06MODTRAN}, with terminals 10 m above ground level and 23 km clear-weather visibility and 1~km link distance. We compute $\bar{n}_{B_2}$ from the total radiance at a solar elevation of $60\degree$.
We set transmission time $T=100$ s,  mean photons per mode $\bar{n}_s'=10^{-4}$, modulation frequency $W_M=100$ MHz, and signal bandwidth $W=1$ THz for broadband schemes, 1, 2, and 3. 

\begin{figure}[htpb]
    \includegraphics[width=0.48\textwidth]{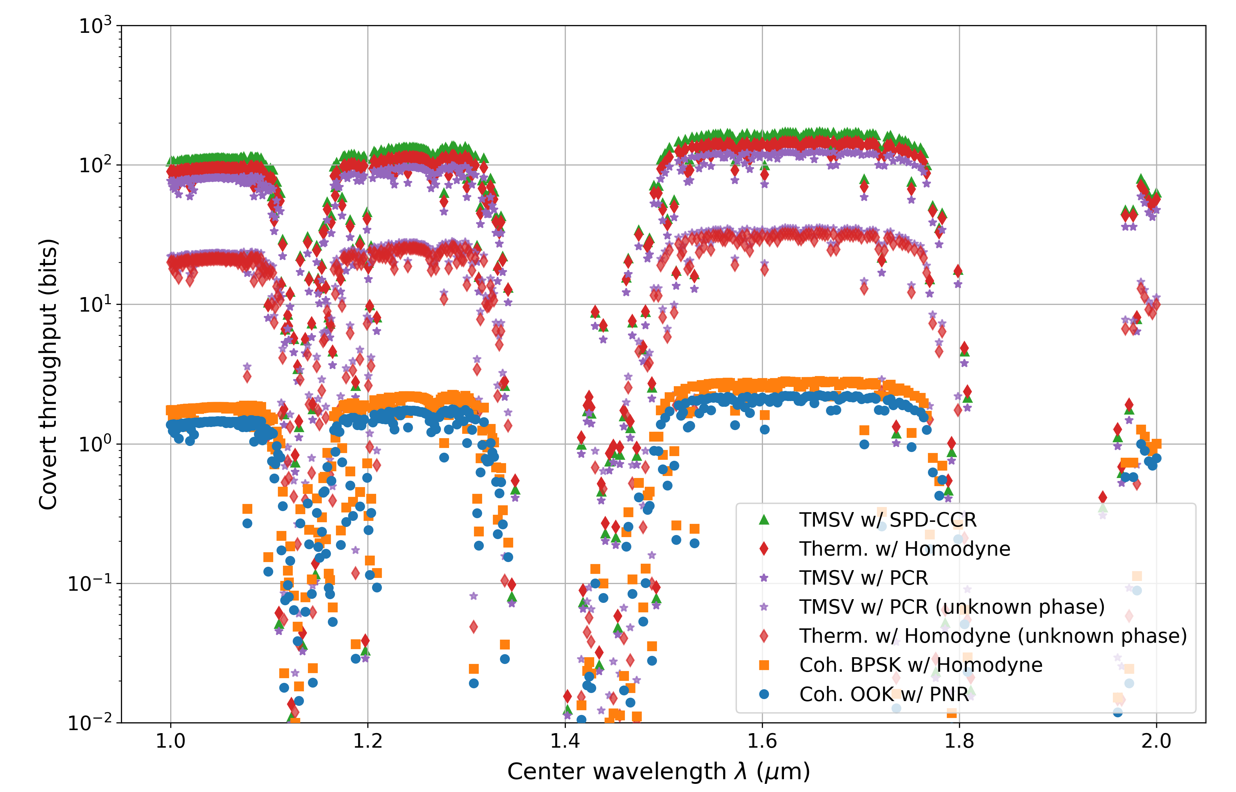}
    \caption[Simulated throughput of the two-way covert communications schemes.]{Simulated throughput of the two-way covert communications schemes described in Section \protect\ref{sec:CCR}.
     \label{fig:RESULTS-CCR}}
\end{figure}

Fig.~\ref{fig:RESULTS-CCR} plots the total covert throughput of all schemes for center wavelengths $\lambda\in[1,2]~\mu$m. The broadband schemes 2 and 3 are phase-sensitive. We consider two scenarios: a) a genie gives Alice the exact phase; and b) phase is stable but unknown in each symbol slot. To estimate the phase in scenario (b), we partition each symbol period into $M_1$-mode phase-sensing and $M_2$-mode data segments, $M=M_1+M_2$, as in \cite{bali25covertsdr-milcom, bali26covertsdr}. We assume optimal phase measurements that achieve quantum Cram\'er-Rao bounds (QCRBs) for schemes 2 and 3. These QCRBs are derived in \cite{Hao_2022}. We maximize the induced channel capacity over $(M_1,M_2)$ using the Monte Carlo simulations for each value of center wavelength $\lambda$. 
Fig.~\ref{fig:RESULTS-CCR} shows that our correlator-based scheme outperforms the others, even when the channel phase is known a priori by schemes 2 and 3. When it is unknown, our scheme performs far better. In \cite{bullock25thesis}, we show that our scheme also performs well with diffuse signaling.
Finally, evaluating the probability $p_s$ of successfully determining the Alice-to-Bob distance $l$ to within $c/W$ is computationally onerous for large $n$. However, we find that, for $(W, T, W_M)$ set as above, $\eta_2=0.8$, $\bar{n}_{B_2}=10^{-5}$, and $l$ known to within $30$ meters prior to transmission (i.e., $l_u-l_l=30$ m), $p_s\geq 1-10^{-20}$.

\section{Conclusion, Discussion, and Future Work} \label{sec:discussion}
We model and analyze a two-way covert communication system based on wide-band covert-sensing ideas \cite{bash19covertsensor, gagatsos19floodlightsensorcleo, gagatsos19floodlightsensor, Hao_2022}. The practicality of these systems is limited by the need for continual precise path length measurement between users. Therefore, we employ narrow-band laser light in our experimental proof-of-concept for two-way covert communication. Furthermore, we offer a solution that attains the wide-band performance gains without the path-length measurement overhead. We will realize it experimentally in the future work.

\section*{Acknowledgment}

We benefited from discussions with Christos N.~Gagatsos, Matthieu R.~Bloch, and Robert A.~Norwood.

\bibliographystyle{IEEEtran}
\bibliography{papers}

\appendices
\section{Proof of Theorem \ref{thm:sparse_signaling}} \label{ap:covertness}
Here, we derive the bound given in Theorem \ref{thm:sparse_signaling}. To use Pinsker's inequality as discussed in Section \ref{sec:covertnesscriterion}, we must evaluate the QRE between the $M$-length block mixed state received by Willie and the $M$-length block innocent thermal state,
\begin{equation} \label{QRE}
\begin{split}
    D(\hat{\rho}_1^{W^M} ||\hat{\rho}_0^{W^M})& =
     \text{Tr}(\hat{R})
\end{split}
\end{equation}
where we have defined an operator, $\hat{R}=\hat{\rho}_1^{W^M}\log\hat{\rho}_1^{W^M}-\hat{\rho}_1^{W^M}\log\hat{\rho}_0^{W^M}$. To estimate the QRE, we need to find the Taylor series expansion of $\hat{R}$ when the Bob's return signal power is equal to 0 , $\bar{n}_T=0$. 

Note that the following analysis represents the input displacement as spread over all $M$ modes rather than localized to a single mode as in Section \ref{subsec:input_signal}, amounting to passive unitary transformations. Willie intercepts the following mixed state for one symbol-slot transmission, 
\begin{equation} \label{mixed}
\begin{split}
   \hat{\rho}_1^{W^M}(\theta)&=(1-\tau)\hat{\rho}_{\text{th}}(\eta_2\bar{n}_{B_2})^{\otimes M}\\&\phantom{=}+\tau\hat{\rho}_{\text{th}}(\bar{n}_W,\alpha_\theta)\otimes\hat{\rho}_{\text{th}}(\bar{n}_W)^{\otimes M-1}\\&= (1-\tau)\hat{\rho}_{\bar{n}_0}^{\otimes M}
   + 
   \tau \hat{\rho}_Q^{(M)}
\end{split}
\end{equation}
where $\hat{\rho}_Q^{(M)} = \frac{1}{4} \left(\rho_{00} +\rho_{01}+\rho_{10}+\rho_{11}\right)$ is the non-innocent QPSK state, and $\hat{\rho}_{\bar{n}_0}=\hat{\rho}_\text{th}(\eta_2\bar{n}_{B_2})$ with $\bar{n}_0=\eta_2\bar{n}_{B_2}$. The individual QPSK modulated states are defined as $\hat{\rho}_{00}\equiv  \hat{\rho}_\text{th}(\bar{n}_W,\alpha_{\theta=0}/ \sqrt{M})^{\otimes M}$, $\hat{\rho}_{01}\equiv  \hat{\rho}_\text{th}(\bar{n}_W,\alpha_{\theta=\pi/2}/\sqrt{M})^{\otimes M}$, $\hat{\rho}_{10}\equiv  \hat{\rho}_\text{th}(\bar{n}_W,\alpha_{\theta=\pi}/\sqrt{M})^{\otimes M}$, $\hat{\rho}_{11}\equiv  \hat{\rho}_\text{th}(\bar{n}_W,\alpha_{\theta=3\pi/2}/\sqrt{M})^{\otimes M}$. 
Recall from Section \ref{subsec:input_signal}, $\bar{n}_T= \bar{n}_\alpha/M + \bar{n}_{S}'$ is the total mean photon number per mode of Bob's transmitted state. We parametrize this by the following ratio, $x\in[0,1]$ such that $\bar{n}_\alpha/M=(1-x)\bar{n}_T$ and $\bar{n}_{S}'=x\bar{n}_T$, allowing the displacement and thermal mean photon number of Willie's input to be expressed as $\alpha_\theta/\sqrt{M}=\sqrt{(1-\eta_2)(1-x)\bar{n}_T}e^{j\theta}$ and $\bar{n}_w=(1-\eta_2)x\bar{n}_T +\eta_2\bar{n}_{B_2}$. 

Let $\hat{\rho}_\tau^{(M)}=\hat{\rho}_1^{W^M}(\theta)$ and $b=\sqrt{\bar{n}_T}$ as differentiation with respect to $\bar{n}_T$ and evaluation at $\bar{n}_T=0$ leads to a singularity. To proceed, we need the first four derivatives of $\hat{\rho}_1^{W^M}(\theta)$ with respect to $b$,
\begin{equation} \label{eq:rho_tau}
\begin{split}
   \frac{\dif^n}{\dif{b}^n}\hat{\rho}_\tau^{(M)} =  
   \frac{\tau}{4}\left( \frac{\dif^n}{\dif{b}^n}\hat{\rho}_{00} + \frac{\dif^n}{\dif{b}^n}\hat{\rho}_{01} + \frac{\dif^n}{\dif{b}^n}\hat{\rho}_{10} + \frac{\dif^n}{\dif{b}^n}\hat{\rho}_{11}\right).
\end{split}
\end{equation}

Let $\rho_\theta^{(M)}=\hat{\rho}_\text{th}(\bar{n}_W,\alpha_{\theta}/ \sqrt{M})^{\otimes M}$. Note that \eqref{eq:rho_tau} involves the summation of derivatives of $\rho_\theta^{(M)}$, therefore we proceed as such, 
\begin{equation} 
\begin{split}
   \frac{\dif}{\dif b}&\rho_\theta^{(M)} 
   =\;\frac{\dif}{\dif b}\bigotimes^M_{k=1}\left( \int_\mathbf{} \dif ^2\beta \frac{\exp[\frac{-|\beta_k-\alpha_\theta|^2}{\bar{n}_W}]}{\pi \bar{n}_W} \ket{\beta_k}\bra{\beta_k} \right)    \\
    =\;&  \int_\mathbf{} \dif ^2\beta_1...\int_\mathbf{} \dif ^2\beta_M\frac{\dif}{\dif b}\left(\prod_{k=1}^M  \frac{\exp[\frac{-|\beta_k-\alpha_\theta|^2}{\bar{n}_W}]}{\pi \bar{n}_W} \right) \\
&\;\;\;\;\;\;\;\;\;\;\;\;\;\;\;\;\;\;\; \;\;\;\;\;\;\;\;\;\;\;\;\;\;\;\;\;\;\;\; \;\;\;\;\;\;\;\;\;\bigotimes^M_{k=1}\ket{\beta_k}\bra{\beta_k}      \\
    =\;&  \int_\mathbf{} \dif ^2\beta_1...\int_\mathbf{} \dif ^2\beta_M\sum_{k=1}^M\left(A^{(1)}_{k,b}\prod_{j=1,\;j\neq k}^M  A^{(0)}_{j,b} \right)
    \\
    &\;\;\;\;\;\;\;\;\;\;\;\;\;\;\;\;\;\;\; \;\;\;\;\;\;\;\;\;\;\;\;\;\;\;\;\;\;\;\; \;\;\;\;\;\;\;\;\;\;\;\;\;\;\;\bigotimes^M_{k=1}\ket{\beta_k}\bra{\beta_k}
\end{split}
\end{equation}
where $A^{(0)}_{k,b}=\frac{1}{\pi \bar{n}_W}\exp[\frac{-|\beta_k-\alpha_\theta|^2}{\bar{n}_W}]=\frac{1}{\pi \bar{n}_W}e_k(b)$, and

\begin{equation} 
\begin{split}
    A^{(1)}_{k,b} &= \frac{\dif}{\dif b}\left(A^{(0)}_{k,b}\right) = \frac{\dif}{\dif b}\left[ \frac{1}{\pi \bar{n}_W}e_k(b) \right] \\
    &= A^{(0)}_{k,b} \left\{ \frac{-2b(1-\eta_2)x}{\bar{n}_W} + \frac{\dif}{\dif b}\left[\frac{-|\beta_k-\alpha_\theta|^2}{\bar{n}_W}\right] \right\}. 
\end{split}
\end{equation}
The derivative in the 2nd term evaluates to 
\begin{equation} 
\begin{split}
\frac{\dif}{\dif b}&\left[\frac{-|\beta_k-\alpha_\theta|^2}{\bar{n}_W}\right]  
=\frac{1}{\bar{n}_W}\Bigg\{\frac{2b(1-\eta_2)x}{ \bar{n}_W}|\beta_k-\alpha_\theta|^2  \\ &-  2b(1-\eta_2)(1-x) + 
2\sqrt{(1-\eta_2)(1-x)}\Re\{\beta_ke^{-j\theta}\}\Bigg\}.
\end{split}
\end{equation}
By defining $B(b)= \frac{1}{\bar{n}_W}\Big(-2(1-\eta_2)b +
\frac{ 2b(1-\eta_2)x}{ \bar{n}_W}|\beta_k-\alpha_\theta|^2+ \sqrt{(1-\eta_2)(1-x)}\left(\beta_ke^{-j\theta}  + \beta_k^*e^{j\theta} \right)\Big)$, we obtain,
\begin{equation} 
\begin{split}
\frac{\dif}{\dif b}\rho_\theta^{(M)} 
&=
 \int_\mathbf{} \dif ^2\beta_1...\int_\mathbf{} \dif ^2\beta_M\\
 &\quad\quad\quad\bigtimes\sum_{k=1}^MB(b)\prod_{j=1}^M  A^{(0)}_{j,b}\bigotimes^M_{m=1}\ket{\beta_m}\bra{\beta_m} \\
&=
\sum_{k=1}^M   
2c_1 \frac{xb}{\bar{n}_W^2}\big(\hat{a}_k -\alpha_\theta\big) \hat{\rho}^{(M)}_\theta  \big( \hat{a}_k^\dagger -\alpha^*(b)\big)\\
&\;\;\;\; +  
\frac{1}{\bar{n}_W}\big(c_2 e^{-j\theta}\hat{a}_k -c_1b\big)\hat{\rho}^{(M)}_\theta \\
&\;\;\;\;+ 
\frac{1}{\bar{n}_W}\hat{\rho}^{(M)}_\theta \big(c_2 e^{j\theta}\hat{a}_k^\dagger -c_1b\big)  
\end{split}
\end{equation}
where we use the fact that $\hat{\rho}^{(M)}_\theta = \int_\mathbf{} \dif ^2\beta_1...\int_\mathbf{} \dif ^2\beta_M \bigotimes^M_{j=1}A^{(0)}_{j,b}  \ket{\beta_j}\bra{\beta_j}$.
Here,  $c_1=1-\eta_2$, $c_2=\sqrt{(1-\eta_2)(1-x)}$. Setting $b=0$ yields,
\begin{equation}\label{eq:rho_theta_1}
\begin{split}
     \frac{\dif}{\dif b}\hat{\rho}_\theta^{(M)}\bigg\vert _{b=0}
     &= 
\sum_{m=1}^M \frac{c_2}{\bar{n}_0} \bigg\{ e^{-j\theta}\hat{a}_m \hat{\rho}_{\bar{n}_0}^{\otimes M} + e^{j\theta}\hat{\rho}_{\bar{n}_0}^{\otimes M}  \hat{a}_m^\dagger \bigg\} .
\end{split}
\end{equation}

Moving onto the second derivative of $\hat{\rho}_\theta^{(M)}$ with respect to $b$,        
\begin{equation} 
\begin{split}
\frac{\dif^2}{\dif b^2}&\rho_\theta^{(M)}
= 
\sum_{m=1}^M 
\frac{2c_1x}{\bar{n}_W^2}\big(\hat{a}_m -\alpha_\theta\big) \hat{\rho}^{(M)}_\theta  \big( \hat{a}_m^\dagger -\alpha^*(b)\big)\\
&\;\;\;\;
- \frac{2(2c_1x)^2b^2}{\bar{n}^3(b)}\big(\hat{a}_m -\alpha_\theta\big) \hat{\rho}^{(M)}_\theta  \big( \hat{a}_m^\dagger -\alpha^*(b)\big)  \\
&\;\;\;\;
- \frac{2c_1c_2xb}{\bar{n}_W^2} \Big[e^{j\theta} \hat{\rho}^{(M)}_\theta  \big( \hat{a}_m^\dagger -\alpha^*(b)\big)   \\
&\;\;\;\;\;\;\;\;\;\;\;\;\;\;\;\;\;\;\;\;\;\;\;\;\;\;
+ e^{-j\theta}\big(\hat{a}_m -\alpha_\theta\big) \hat{\rho}^{(M)}_\theta   \Big] \\
&\;\;\;\;
+\frac{2c_1xb}{\bar{n}_W^2}\big(\hat{a}_m -\alpha_\theta\big) \frac{\dif}{\dif b}\hat{\rho}_\theta^{(M)}  \big( \hat{a}_m^\dagger -\alpha^*(b)\big)  \\
&\;\;\; \;
- \frac{2c_1xb}{\bar{n}_W^2}\Big[\left(c_2e^{-j\theta}\hat{a}_m-c_1b\right)\hat{\rho}_\theta^{(M)}\\
&\;\;\;\;\;\;\;\;\;\;\;\;\;\;\;\;\;\;\;\;\;\;\;\;\;\; + \hat{\rho}_\theta^{(M)}\left(c_2e^{j\theta}\hat{a}^\dagger_m-c_1b\right) \Big]\\
&\;\;\; \;
+ \frac{1}{\bar{n}_W}\Big[\left(c_2e^{-j\theta}\hat{a}_m-c_1b\right)\frac{\dif}{\dif b}\hat{\rho}_\theta^{(M)} +\\
&\;\;\;\;\;\;\;\;\;\;\;\;\;\;\;\;\;\;\;\;\;\;\;\;\;\;
\frac{\dif}{\dif b}\hat{\rho}_\theta^{(M)}\left(c_2e^{j\theta}\hat{a}^\dagger_m-c_1b\right) \Big]\\
&\;\;\;\;
-2\frac{c_1}{\bar{n}_W}\hat{\rho}_\theta^{M} 
\end{split}
\end{equation}
Setting $b=0$ yields
\begin{equation} \label{eq:rho_theta_2}
\begin{split}
&\frac{\dif ^2}{\dif b^2}\rho_\theta^{(M)} \bigg\vert _{b=0}
=
\sum_{m=1}^M \frac{1}{\bar{n}_0} \bigg\{ \frac{2c_1x}{\bar{n}_0}\hat{a}_m \hat{\rho}_{\bar{n}_0}^{\otimes M}   \hat{a}_m^\dagger 
- 2c_1\hat{\rho}_{\bar{n}_0}^{\otimes M}
  \\
&\;\;\;\;\;\;\;\;
+ c_2 \Big[ e^{-j\theta}\hat{a}_m \frac{\dif }{\dif b}\hat{\rho}_\theta^{(M)}\bigg\vert _{b=0} + e^{j\theta} \frac{\dif }{\dif b}\hat{\rho}_\theta^{(M)}\bigg\vert _{b=0} \hat{a}_m^\dagger \Big] \bigg\}\\
&\;\;\;\;=\sum_{{m_1},m_2=1}^M  \frac{2c_1x}{ \bar{n}_0^{2}}\hat{a}_{m_1} \hat{\rho}_{\bar{n}_0}^{\otimes M}   \hat{a}_{m_1}^\dagger 
- \frac{2c_1}{\bar{n}_0}\hat{\rho}_{\bar{n}_0}^{\otimes M}
  \\
&\;\;\;\;\;\;\;\;
+ \frac{c_2^2}{\bar{n}_0^{2}}        \Big[ e^{-2j\theta}\hat{a}_{m_1}\hat{a}_{m_2}\hat{\rho}_{\bar{n}_0}^{\otimes M} + e^{2j\theta}\hat{\rho}_{\bar{n}_0}^{\otimes M} \hat{a}_{m_2}^\dagger\hat{a}_{m_1}^\dagger\\
&\;\;\;\;\;\;\;\;\;\;\;\;\;\;\;\;\;\;\;\;\;
+ 2\hat{a}_{m_2}\hat{\rho}_{\bar{n}_0}^{\otimes M}\hat{a}_{m_1}^\dagger \Big].\\
\end{split}
\end{equation}

The third derivative of $\hat{\rho}_\theta^{(M)}$ with respect to $b$ is as follows, 
\begin{equation} 
\begin{split}
\frac{\dif ^3}{\dif b^3}&\rho_\theta^{(M)}
= 
\sum_{m=1}^M 
\frac{-6(2c_1x)^2b}{\bar{n}^3(b)} \\
&\;\;\;\;\;\;\;\;\;\;\bigtimes\big(\hat{a}_m -\alpha_\theta\big) \hat{\rho}^{(M)}_\theta  \big( \hat{a}_m^\dagger -\alpha^*(b)\big)\\
&\;\;\;\;\;\;\;
+ 
\frac{4c_1x}{\bar{n}_W^2}\big(\hat{a}_m -\alpha_\theta\big) \frac{\dif }{\dif b}\hat{\rho}_\theta^{(M)}  \big( \hat{a}_m^\dagger -\alpha^*(b)\big)\\
&\;\;\;\;\;\;\;
- \frac{4c_1c_2x}{\bar{n}_W^2} \bigg[e^{j\theta} \hat{\rho}^{(M)}_\theta  \big( \hat{a}_m^\dagger -\alpha^*(b)\big)\\&
\;\;\;\;\;\;\;\;\;\;\;\;\;\;\;\;\;\;\;\;\;\;\;\;\;\;\;\;\;\;\;\;
+ e^{-j\theta}\big(\hat{a}_m -\alpha_\theta\big) \hat{\rho}^{(M)}_\theta   \bigg] \\
&\;\;\;\;\;\;\;
+\frac{4c_1xb}{\bar{n}_W^2}\left( c_2^2+2c_1\right)\hat{\rho}_\theta^{(M)}  \\
&\;\;\;\;\;\;\;
+ 
\frac{2c_1xb}{\bar{n}_W^2}\big(\hat{a}_m -\alpha_\theta\big) \frac{\dif ^2}{\dif b^2}\hat{\rho}_\theta^{(M)}  \big( \hat{a}_m^\dagger -\alpha^*(b)\big)\\
&\;\;\;\;\;\;\;
-
\frac{4c_1c_2xb}{\bar{n}_W^2} \bigg[e^{j\theta} \frac{\dif }{\dif b}\hat{\rho}^{(M)}_\theta  \big( \hat{a}_m^\dagger -\alpha^*(b)\big) \\&
\;\;\;\;\;\;\;\;\;\;\;\;\;\;\;\;\;\;\;\;\;\;\;\;\;\;\;\;\;\;\;\;         + e^{-j\theta}\big(\hat{a}_m 
-\alpha_\theta\big) \frac{\dif }{\dif b}\hat{\rho}^{(M)}_\theta   \bigg] \\
&\;\;\; \;\;\;\;
- \frac{2c_1x}{\bar{n}_W^2}\bigg[\left(c_2e^{-j\theta}\hat{a}_m-c_1b\right)\hat{\rho}_\theta^{(M)}\\&
\;\;\;\;\;\;\;\;\;\;\;\;\;\;\;\;\;\;\;\;\;\;\;\;\;\;\;\;\;\;\;\;
 + \hat{\rho}_\theta^{(M)}\left(c_2e^{j\theta}\hat{a}^\dagger_m-c_1b\right) \bigg]\\
&\;\;\; \;\;\;\;
- \frac{4c_1xb}{\bar{n}_W^2}\bigg[\left(c_2e^{-j\theta}\hat{a}_m-c_1b\right)\frac{\dif }{\dif b}\hat{\rho}_\theta^{(M)}\\&
\;\;\;\;\;\;\;\;\;\;\;\;\;\;\;\;\;\;\;\;\;\;\;\;\;\;\;\;\;\;\;\;
 + \frac{\dif }{\dif b}\hat{\rho}_\theta^{(M)}\left(c_2e^{j\theta}\hat{a}^\dagger_m-c_1b\right) \bigg]\\
&\;\;\; \;\;\;\;
+ \frac{1}{\bar{n}_W}\bigg[\left(c_2e^{-j\theta}\hat{a}_m-c_1b\right)\frac{\dif ^2}{\dif b^2}\hat{\rho}_\theta^{(M)} \\&
\;\;\;\;\;\;\;\;\;\;\;\;\;\;\;\;\;\;\;\;\;\;\;\;\;\;\;\;\;\;\;\;
+ \frac{\dif ^2}{\dif b^2}\hat{\rho}_\theta^{(M)}\left(c_2e^{j\theta}\hat{a}^\dagger_m-c_1b\right) \bigg]\\
&\;\;\;\;\;\;\;
-\frac{4c_1}{\bar{n}_W}\frac{\dif }{\dif b}\hat{\rho}_\theta^{M} + o(b).\\
\end{split}
\end{equation}\\
Note that higher order terms of $b$ have zero contribution when $b=0$. Evaluating at $b=0$,
\begin{equation} \label{eq:rho_theta_3}
\begin{split}
\frac{\dif ^3}{\dif b^3}&\rho_\theta^{(M)}\bigg\vert _{b=0}
=
\sum_{m=1}^M
\frac{4c_1x}{\bar{n}_0^2}\hat{a}_m \frac{\dif \hat{\rho}_\theta^{(M)}}{\dif b} \bigg\vert _{b=0} \hat{a}_m^\dagger\\
&\;\;\;\;-
\frac{6c_1c_2x}{\bar{n}_0^{2}}\left[  e^{-j\theta} \hat{a}_m \hat{\rho}_{\bar{n}_0}^{\bigotimes M} + e^{j\theta}\hat{\rho}_{\bar{n}_0}^{\bigotimes M}\hat{a}_m^\dagger\right] \\
&\;\;\;\;+ 
\frac{c_2}{\bar{n}_0}\left[e^{-j\theta} \hat{a}_m \frac{\dif ^2\hat{\rho}_\theta^{(M)}}{\dif b^2}\bigg\vert _{b=0} + e^{j\theta}\frac{\dif ^2\hat{\rho}_\theta^{(M)}}{\dif b^2}\bigg\vert _{b=0}\hat{a}_m^\dagger \right]\\
&\;\;\;\;-
\frac{4c_1}{\bar{n}_0} \frac{\dif \hat{\rho}_\theta^{(M)}}{\dif b} \bigg\vert _{b=0}\\
\end{split}
\end{equation}

Moving onto the fourth derivative of $\hat{\rho}_\theta^{(M)}$ with respect to $b$, 
\begin{equation} 
\begin{split}
\frac{\dif ^4}{\dif b^4}&\rho_\theta^{(M)}
= 
\sum_{m=1}^M 
\frac{-24c_1^2x^2}{\bar{n}^3(b)}\big(\hat{a}_m -\alpha_\theta\big) \hat{\rho}^{(M)}_\theta  \big( \hat{a}_m^\dagger -\alpha^*(b)\big)\\
&
+ 
\frac{6c_1x}{\bar{n}_W^2}\big(\hat{a}_m -\alpha_\theta\big) \frac{\dif ^2}{\dif b^2}\hat{\rho}_\theta^{(M)}  \big( \hat{a}_m^\dagger -\alpha^*(b)\big)\\
&
- \frac{12c_1c_2x}{\bar{n}_W^2} \bigg[e^{j\theta} \frac{\dif }{\dif b}\hat{\rho}^{(M)}_\theta  \big( \hat{a}_m^\dagger -\alpha^*(b)\big)          \\& \;\;\;\;\;\;\;\;\;\;\;\;\;\;\;\;\;\;\;\;\;\;\;
+ e^{-j\theta}\big(\hat{a}_m -\alpha_\theta\big) \frac{\dif }{\dif b}\hat{\rho}^{(M)}_\theta   \bigg] \\
&
+\frac{12c_1x}{\bar{n}_W^2}\left( c_2^2+c_1\right)\hat{\rho}_\theta^{(M)}  \\
&
- \frac{6c_1x}{\bar{n}_W^2}\bigg[\left(c_2e^{-j\theta}\hat{a}_m-c_1b\right)\frac{\dif }{\dif b}\hat{\rho}_\theta^{(M)} \\& \;\;\;\;\;\;\;\;\;\;\;\;\;\;\;\;\;\;\;\;\;\;\;
+ \frac{\dif }{\dif b}\hat{\rho}_\theta^{(M)}\left(c_2e^{j\theta}\hat{a}^\dagger_m-c_1b\right) \bigg]\\
&
+ \frac{1}{\bar{n}_W}\bigg[\left(c_2e^{-j\theta}\hat{a}_m-c_1b\right)\frac{\dif ^3}{\dif b^3}\hat{\rho}_\theta^{(M)} \\& \;\;\;\;\;\;\;\;\;\;\;\;\;\;\;\;\;\;\;\;\;\;\;
+ \frac{\dif ^3}{\dif b^3}\hat{\rho}_\theta^{(M)}\left(c_2e^{j\theta}\hat{a}^\dagger_m-c_1b\right) \bigg]\\
&
-\frac{6c_1}{\bar{n}_W}\frac{\dif ^2}{\dif b^2}\hat{\rho}_\theta^{M} + o(b). \\
\end{split}
\end{equation}
Setting $b=0$ yields,
\begin{equation} \label{eq:rho_theta_4}
\begin{split}
\frac{\dif ^4}{\dif b^4}&\rho_\theta^{(M)}\bigg\vert_{b=0}
= 
\sum_{m=1}^M 
\frac{-24c_1^2x^2}{\bar{n}_0^3}\hat{a}_m \hat{\rho}^{\bigotimes M}_{\hat{n}_0}  \hat{a}_m^\dagger\\
&
+ 
\frac{6c_1x}{\bar{n}_0^2}\hat{a}_m\frac{\dif ^2}{\dif b^2}\hat{\rho}_\theta^{(M)}\bigg\vert_{b=0}   \hat{a}_m^\dagger\\
&
- \frac{18c_1c_2x}{\bar{n}_0^2} \left[e^{j\theta} \frac{\dif }{\dif b}\hat{\rho}^{(M)}_\theta\bigg\vert_{b=0}  \hat{a}_m^\dagger          + e^{-j\theta}\hat{a}_m \frac{\dif }{\dif b}\hat{\rho}^{(M)}_\theta\bigg\vert_{b=0}   \right] \\
&
+\frac{12c_1x}{\bar{n}_0^2}\left( c_2^2+c_1\right)\hat{\rho}^{\bigotimes M}_{\hat{n}_0}  \\
&
+ \frac{c2}{\bar{n}_0}\left[e^{-j\theta}\hat{a}_m\frac{\dif ^3}{\dif b^3}\hat{\rho}_\theta^{(M)}\bigg\vert_{b=0} + e^{j\theta}\frac{\dif ^3}{\dif b^3}\hat{\rho}_\theta^{(M)}\bigg\vert_{b=0}\hat{a}^\dagger_m \right]\\
&
-\frac{6c_1}{\bar{n}_0}\frac{\dif ^2}{\dif b^2}\hat{\rho}_\theta^{M}\bigg\vert_{b=0}. \\ 
\end{split}
\end{equation}

Using \eqref{eq:rho_theta_1}, \eqref{eq:rho_theta_2}, \eqref{eq:rho_theta_3}, and \eqref{eq:rho_theta_4}, we proceed with calculating the first four derivatives of $\hat{\rho}_\tau^{(M)}$ as expressed in \eqref{eq:rho_tau}. The first derivative of $\hat{\rho}_\tau^{(M)}$ with respect to $b$ evaluated at $b=0$ is 
\begin{equation} \label{eq:rho_tau_1}
\begin{split}
   \frac{\dif}{\dif{b}}\hat{\rho}_\tau^{(M)} \bigg\vert _{b=0}&=  
   \frac{\tau}{4}\Bigg( \frac{\dif}{\dif{b}}\rho_{00}\bigg\vert _{b=0} + \frac{\dif}{\dif{b}}\rho_{01}\bigg\vert _{b=0} \\
   &\;\;\;\;\;\;\;\;+ \frac{\dif}{\dif{b}}\rho_{10}\bigg\vert _{b=0} + \frac{\dif}{\dif{b}}\rho_{11}\bigg\vert _{b=0}\Bigg) \\
   &=
   \frac{\tau}{4}\frac{c_2}{\bar{n}_0}\sum_{m=1}^M
   \hat{a}_m \hat{\rho}_{\bar{n}_0}^{\otimes M} + \hat{\rho}_{\bar{n}_0}^{\otimes M}  \hat{a}_m^\dagger \\
   &\;\;\;\;\;\;\;\;\;\;\;\;\;-
  \left( j\hat{a}_m \hat{\rho}_{\bar{n}_0}^{\otimes M} - j\hat{\rho}_{\bar{n}_0}^{\otimes M}  \hat{a}_m^\dagger \right) \\
   &\;\;\;\;\;\;\;\;\;\;\;\;\;- 
   \left( \hat{a}_m \hat{\rho}_{\bar{n}_0}^{\otimes M} + \hat{\rho}_{\bar{n}_0}^{\otimes M}  \hat{a}_m^\dagger \right) \\
   &\;\;\;\;\;\;\;\;\;\;\;\;\;+ 
   \left( j\hat{a}_m \hat{\rho}_{\bar{n}_0}^{\otimes M} - j\hat{\rho}_{\bar{n}_0}^{\otimes M}  \hat{a}_m^\dagger \right) \\
   &=0
\end{split}
\end{equation}
where the presence of $e^{\pm qj\frac{\pi}{2}}$ in each term of \eqref{eq:rho_theta_1} causes all terms to cancel when summing over $q\in \{0,1,2,3\}$. 

The second derivative of $\hat{\rho}_\tau^{(M)}$ with respect to $b$ evaluated at $b=0$ is 
\begin{equation} \label{eq:rho_tau_2}
\begin{split}
   \frac{\dif^2}{\dif{b}^2}\hat{\rho}_\tau^{(M)} \bigg\vert _{b=0}&=\tau\sum_{{m_1},m_2=1}^M  \frac{2c_1x}{ \bar{n}_0^{2}}\hat{a}_{m_1} \hat{\rho}_{\bar{n}_0}^{\otimes M}   \hat{a}_{m_1}^\dagger \\
&- \frac{2c_1}{\bar{n}_0}\hat{\rho}_{\bar{n}_0}^{\otimes M}  
+ 2\frac{c_2^2}{\bar{n}_0^{2}}        \hat{a}_{m_2}\hat{\rho}_{\bar{n}_0}^{\otimes M}\hat{a}_{m_1}^\dagger  \\
&
+ \frac{c_2^2}{4\bar{n}_0^{2}}        \left[ \hat{a}_{m_1}\hat{a}_{m_2}\hat{\rho}_{\bar{n}_0}^{\otimes M} + \hat{\rho}_{\bar{n}_0}^{\otimes M} \hat{a}_{m_2}^\dagger\hat{a}_{m_1}^\dagger \right]\\
&- \frac{c_2^2}{4\bar{n}_0^{2}}        \left[ \hat{a}_{m_1}\hat{a}_{m_2}\hat{\rho}_{\bar{n}_0}^{\otimes M} + \hat{\rho}_{\bar{n}_0}^{\otimes M} \hat{a}_{m_2}^\dagger\hat{a}_{m_1}^\dagger \right]\\
&+ \frac{c_2^2}{4\bar{n}_0^{2}}        \left[ \hat{a}_{m_1}\hat{a}_{m_2}\hat{\rho}_{\bar{n}_0}^{\otimes M} + \hat{\rho}_{\bar{n}_0}^{\otimes M} \hat{a}_{m_2}^\dagger\hat{a}_{m_1}^\dagger \right]\\
&- \frac{c_2^2}{4\bar{n}_0^{2}}        \left[ \hat{a}_{m_1}\hat{a}_{m_2}\hat{\rho}_{\bar{n}_0}^{\otimes M} + \hat{\rho}_{\bar{n}_0}^{\otimes M} \hat{a}_{m_2}^\dagger\hat{a}_{m_1}^\dagger \right]\\
&=\tau\sum_{{m_1},m_2=1}^M  \frac{2c_1x}{ \bar{n}_0^{2}}\hat{a}_{m_1} \hat{\rho}_{\bar{n}_0}^{\otimes M}   \hat{a}_{m_1}^\dagger \\
&- \frac{2c_1}{\bar{n}_0}\hat{\rho}_{\bar{n}_0}^{\otimes M}  
+ 2\frac{c_2^2}{\bar{n}_0^{2}}        \hat{a}_{m_2}\hat{\rho}_{\bar{n}_0}^{\otimes M}\hat{a}_{m_1}^\dagger.  \\
\end{split}
\end{equation}

The third derivative of $\hat{\rho}_\tau^{(M)}$ with respect to $b$ evaluated at $b=0$ is 
\begin{equation} \label{eq:rho_tau_3}
\begin{split}
   &\frac{\dif^3}{\dif{b}^3}\hat{\rho}_\tau^{(M)} \bigg\vert _{b=0}
   =
\frac{\tau}{4}\sum_{m=1}^M\sum_{q=0}^3
\frac{4c_1x}{\bar{n}_0^2}\hat{a}_m \frac{d\hat{\rho}_\theta^{(M)}}{db} \bigg\vert _{b=0} \hat{a}_m^\dagger\\
&\;\;\;\;-
\frac{6c_1c_2x}{\bar{n}_0^{2}}\left[  e^{-qj\frac{\pi}{2}} \hat{a}_m \hat{\rho}_{\bar{n}_0}^{\bigotimes M} + e^{qj\frac{\pi}{2}}\hat{\rho}_{\bar{n}_0}^{\bigotimes M}\hat{a}_m^\dagger\right] \\
&\;\;\;\;+ 
\frac{c_2}{\bar{n}_0}\left[e^{-qj\frac{\pi}{2}} \hat{a}_m \frac{\dif ^2}{\dif b^2}\hat{\rho}_\theta^{(M)}\bigg\vert _{b=0} + e^{qj\frac{\pi}{2}}\frac{\dif ^2}{\dif b^2}\hat{\rho}_\theta^{(M)}\bigg\vert _{b=0}\hat{a}_m^\dagger \right]\\
&\;\;\;\;-
\frac{4c_1}{\bar{n}_0} \frac{\dif }{\dif b}\hat{\rho}_\theta^{(M)} \bigg\vert _{b=0}\\ 
\end{split}
\end{equation}

Now we move onto the next derivative,
\begin{equation} 
\begin{split}
&\frac{\dif^4}{\dif b^4}\rho_\tau^{(M)}\bigg\vert_{b=0} 
= \frac{\tau}{4}\sum_{\theta}\sum_{m=1}^M 
\frac{-24c_1^2x^2}{\bar{n}_0^3}\hat{a}_m \hat{\rho}^{\bigotimes M}_{\hat{n}_0}  \hat{a}_m^\dagger\\
&\;\;\;
+ 
\frac{6c_1x}{\bar{n}_0^2}\hat{a}_m\frac{\dif ^2}{\dif b^2}\hat{\rho}_Q^{(M)}\bigg\vert_{b=0}   \hat{a}_m^\dagger-\frac{6c_1}{\bar{n}_0}\frac{\dif ^2}{\dif b^2}\hat{\rho}_\theta^{M}\bigg\vert_{b=0}\\
&\;\;\;
- \frac{36c_1c_2^2x}{\bar{n}_0^2} \left[e^{j\theta} \frac{\dif }{\dif b}\hat{\rho}^{(M)}_\theta\bigg\vert_{b=0}  \hat{a}_m^\dagger          + e^{-j\theta}\hat{a}_m \frac{\dif }{\dif b}\hat{\rho}^{(M)}_\theta\bigg\vert_{b=0}   \right] \\
&\;\;\;
+\frac{12c_1x}{\bar{n}_0^2}\left( c_2^2+c_1\right)\hat{\rho}_\theta^{(M)}\bigg\vert_{b=0}  \\
&\;\;\;
+ \frac{c2}{\bar{n}_0}\left[e^{-j\theta}\hat{a}_m\frac{\dif ^3}{\dif b^3}\hat{\rho}_\theta^{(M)}\bigg\vert_{b=0} + e^{j\theta}\frac{\dif ^3}{\dif b^3}\hat{\rho}_\theta^{(M)}\bigg\vert_{b=0}\hat{a}^\dagger_m \right]\\
&\;\;\;
\end{split}
\end{equation}
Plugging in the expressions for the derivatives, we obtain
\begin{equation} \label{eq:rho_tau_4}
    \begin{split}
        \frac{\dif^4}{\dif b^4}\rho_\tau^{(M)}&\bigg\vert_{b=0}=\frac{-24c_1^2x^2}{\bar{n}_0^3}\sum_{m=1}^M 
\hat{a}_m \hat{\rho}^{\bigotimes M}_{\hat{n}_0}  \hat{a}_m^\dagger\\
&\;\;
+ 
\frac{12c_1^2x^2}{\bar{n}_0^4}\sum_{m_1,m_2=1}^M\hat{a}_{m_1}\hat{a}_{m_2}\hat{\rho}^{\bigotimes M}_{\hat{n}_0}\hat{a}_{m_2}^\dagger\hat{a}_{m_1}^\dagger\\
&\;\;
-
\frac{12c_1^2xM}{\bar{n}_0^3}\sum_{m_1=1}^M\hat{a}_{m_1}\hat{\rho}^{\bigotimes M}_{\hat{n}_0}\hat{a}_{m_1}^\dagger\\
&\;\;
+
\frac{12c_1c_2^2x}{\bar{n}_0^4}\sum_{m_1,m_2,m_3=1}^M\hat{a}_{m_1}\hat{a}_{m_3}\hat{\rho}^{\bigotimes M}_{\hat{n}_0}\hat{a}_{m_2}^\dagger\hat{a}_{m_1}^\dagger\\
&\;\;
-
\frac{36c_1c_2^2x}{\bar{n}_0^3}\sum_{m_1,m_2=1}^M 
\hat{a}_{m_2} \hat{\rho}^{\bigotimes M}_{\hat{n}_0}  \hat{a}_{m_1}^\dagger\\
&\;\;
+ 
\frac{12c_1x}{\bar{n}_0^2}\left( c_2^2+c_1\right)\hat{\rho}^{\bigotimes M}_{\hat{n}_0}\\
&\;\;
- 
\frac{12c_1^2xM}{\bar{n}_0^4}\sum_{m_2=1}^M\hat{a}_{m_2}\hat{\rho}^{\bigotimes M}_{\hat{n}_0}\hat{a}_{m_2}^\dagger\\
&\;\;
+
\frac{12c_1^2M^2}{\bar{n}_0^2}\hat{\rho}^{\bigotimes M}_{\hat{n}_0}\\
&\;\;
-
\frac{24c_1c_2^2M}{\bar{n}_0^4}\sum_{m_2,m_3=1}^M\hat{a}_{m_3}\hat{\rho}^{\bigotimes M}_{\hat{n}_0}\hat{a}_{m_2}^\dagger\\
&\;\;
+
\frac{6c_2^4}{\bar{n}_0^4}\sum_{m_1,m_2,m_3,m_4=1}^M\hat{a}_{m_1}\hat{a}_{m_2}\hat{\rho}^{\bigotimes M}_{\hat{n}_0}\hat{a}_{m_4}^\dagger\hat{a}_{m_3}^\dagger\\
&\;\;
+
\frac{8c_1c_2^2x}{\bar{n}_0^4}\sum_{m_1,m_2,m_3=1}^M\hat{a}_{m_1}\hat{a}_{m_2}\hat{\rho}^{\bigotimes M}_{\hat{n}_0}\hat{a}_{m_3}^\dagger\hat{a}_{m_2}^\dagger\\
&\;\;
-
\frac{12c_1c_2^2x}{\bar{n}_0^4}\sum_{m_1,m_2=1}^M\hat{a}_{m_1}\hat{\rho}^{\bigotimes M}_{\hat{n}_0}\hat{a}_{m_2}^\dagger\\
&\;\;
+
\frac{4c_1c_2^2x}{\bar{n}_0^4}\sum_{m_1,m_2,m_3=1}^M\hat{a}_{m_1}\hat{a}_{m_3}\hat{\rho}^{\bigotimes M}_{\hat{n}_0}\hat{a}_{m_3}^\dagger\hat{a}_{m_2}^\dagger\\
&\;\;
    \end{split}
\end{equation}

Having found expressions for the first four derivatives of $\hat{\rho}_\tau^{(M)}$ with respect to $b$, we proceed with finding the first four terms for the Taylor expansion of $\hat{R}$. We use the following lemmas, where $\hat{A}(t)$ and $\hat{B}(t)$ are non-singular operators  parametrized by $t$, and $\hat{I}$ is the identity operator,

\begin{lem}[{\cite[Th.~6]{haber18matrixexplog}}]
\label{lemma:dlogint}
\begin{equation}
\begin{split}
\frac{\mathrm{d}}{\mathrm{d}t}\ln \hat{A}(t)=\int_0^1\mathrm{d}s\left[s\hat{A}(t)+(1-s)\hat{I}\right]^{-1}\\
\bigtimes\frac{\mathrm{d}\hat{A}(t)}{\mathrm{d}t}\left[s\hat{A}(t)+(1-s)\hat{I}\right]^{-1}.
\end{split}
\end{equation}
\end{lem}

\begin{lem}[{\cite[lemma in Sec.~4]{haber18matrixexplog}}]
\label{lemma:dinv}
\begin{align}
\frac{\mathrm{d}}{\mathrm{d}t}\hat{B}^{-1}(t)=-\hat{B}^{-1}(t)\frac{\mathrm{d}\hat{B}(t)}{\mathrm{d}t}\hat{B}^{-1}(t).
\end{align}
\end{lem}

\subsection{First Term}
The first derivative of $\hat{R}$ with respect to $b$ is found using Lemma \ref{lemma:dlogint},

\begin{equation} 
\begin{split}
\frac{\dif \hat{R}}{\dif b} 
&= 
\hat{\rho}_\tau^{(M)}\int_0^1\dif s\left\{ \hat{\sigma}_1^{-1} \frac{\dif }{\dif b}\hat{\rho}_\tau^{(M)}  \hat{\sigma}_1^{-1}\right\} \\
&+\frac{\dif }{\dif b}\hat{\rho}_\tau^{(M)}\log\hat{\rho}_\tau^{(M)} 
- \frac{\dif }{\dif b}\hat{\rho}_\tau^{(M)}\log\hat{\rho}_{\bar{n}_0}^{\otimes M}.
\end{split}
\end{equation}
where $\hat{\sigma}_1=(s\hat{\rho}_\tau^{(M)}+(1-s)\hat{I})$. 
Setting $b=0$ yields,

\begin{equation} 
\begin{split}
\frac{\dif \hat{R}}{\dif b}\bigg\vert _{b=0} 
&= 
\hat{\rho}_{\bar{n}_0}^{\otimes M}\int_0^1\dif s\left\{ \hat{\sigma}_0^{-1}\left( \frac{\dif }{\dif b}\hat{\rho}_\tau^{(M)}\bigg\vert _{b=0} \right) \hat{\sigma}_0^{-1}\right\} \\
&+\frac{\dif }{\dif b}\hat{\rho}_\tau^{(M)}\bigg\vert _{b=0}\log\hat{\rho}_{\bar{n}_0}^{\otimes M} 
- \frac{\dif }{\dif b}\hat{\rho}_\tau^{(M)}\bigg\vert _{b=0}\log\hat{\rho}_{\bar{n}_0}^{\otimes M} \\
&=0
\end{split}
\end{equation}
where $\hat{\sigma}_0= (s\hat{\rho}_{\bar{n}_0}^{\otimes M}+(1-s)\hat{I})$.
The first term is equal to zero as $\frac{\dif}{\dif{b}}\hat{\rho}_\tau^{(M)} \big\vert _{b=0}=0$ and the final two terms cancel.

\subsection{Second Term}
The second derivative of $\hat{R}$ with respect to $b$ is found using Lemma \ref{lemma:dlogint} and Lemma \ref{lemma:dinv},

\begin{equation} 
\begin{split}
&\frac{\dif ^2\hat{R}}{\dif b^2} = 
\frac{\dif ^2}{\dif b^2}\hat{\rho}_\tau^{(M)}\log\hat{\rho}_\tau^{(M)}
- \frac{\dif ^2}{\dif b^2}\hat{\rho}_\tau^{(M)} log\;\hat{\rho}_{\bar{n}_0}^{\otimes M}\\
&\;\;+
2\frac{\dif }{\dif b}\hat{\rho}_\tau^{(M)}\int_0^1\dif s\left\{ \hat{\sigma}_0^{-1}\left( \frac{\dif }{\dif b}\hat{\rho}_\tau^{(M)} \right) \hat{\sigma}_0^{-1}\right\} \\
&\;\;+ 
\hat{\rho}_\tau^{(M)}\int_0^1\dif s\left\{ \hat{\sigma}_0^{-1}\left( \frac{\dif ^2}{\dif b^2}\hat{\rho}_\tau^{(M)} \right) \hat{\sigma}_0^{-1}\right\}\\
&\;\;-
2\hat{\rho}_\tau^{(M)}\int_0^1\dif s \left\{s \hat{\sigma}_0^{-1}\left( \frac{\dif }{\dif b}\hat{\rho}_\tau^{(M)} \right) \hat{\sigma}_0^{-1} \left( \frac{\dif }{\dif b}\hat{\rho}_\tau^{(M)} \right) \hat{\sigma}_0^{-1} \right\} \\
\end{split}
\end{equation}
Setting $b=0$ causes the first two terms to cancel, while the third and fith terms evaluate to zero due to $\frac{\dif}{\dif{b}}\hat{\rho}_\tau^{(M)} \big\vert _{b=0}=0$, resulting in, 
\begin{equation}\label{eq:R_2}
\begin{split}
&\frac{\dif ^2\hat{R}}{\dif b^2}\bigg\vert _{b=0} = 
\hat{\rho}_{\bar{n}_0}^{\otimes M}\int_0^1\dif s\left\{ \hat{\sigma}_0^{-1}\left( \frac{\dif ^2}{\dif b^2}\hat{\rho}_\tau^{(M)}\bigg\vert _{b=0} \right) \hat{\sigma}_0^{-1}\right\}\\
&=
2\tau\bar{n}_0^{-2}c_1x\sum_{m=1}^M\hat{\rho}_{\bar{n}_0}^{\otimes M}\int_0^1\dif s\left\{ \hat{\sigma}_0^{-1}\left( \hat{a}_m \hat{\rho}_{\bar{n}_0}^{\otimes M}   \hat{a}_m^\dagger  \right) \hat{\sigma}_0^{-1}\right\}\\
&-
2\tau\bar{n}_0^{-1}c_1\sum_{m=1}^M\hat{\rho}_{\bar{n}_0}^{\otimes M}\int_0^1\dif s\left\{ \hat{\sigma}_0^{-1}\left( \hat{\rho}_{\bar{n}_0}^{\otimes M} \right) \hat{\sigma}_0^{-1}\right\}\\
&+   
 2\tau\bar{n}_0^{-2}c_2^2  \sum_{m_1,m_2=1}^M \hat{\rho}_{\bar{n}_0}^{\otimes M}\int_0^1\dif s\left\{ \hat{\sigma}_0^{-1}\left(\hat{a}_{m_1}\hat{\rho}_{\bar{n}_0}^{\otimes M} \hat{a}_{m_2}^\dagger \right) \hat{\sigma}_0^{-1}\right\} \\
\end{split}
\end{equation}

Focusing on the 1st term of \eqref{eq:R_2}, note that $\hat{a}_m \hat{\rho}_{\bar{n}_0}^{\otimes M}   \hat{a}_m^\dagger = \sum_{k_1,...,k_M=0}^\infty \frac{\bar{n}_0}{\bar{n}_0 + 1}(k_m+1)\prod_{i=1}^Mt_{k_i} \bigotimes_{i=1}^M\ket{k_i}\bra{k_i}$ where $t_{k_i}=\frac{\bar{n}_0^{k_i}}{(1+\bar{n}_0)^{k_i+1}}$. Thus, 

\begin{equation}
\begin{split}
&\int_0^1\dif s\left\{ \hat{\sigma}_0^{-1}\left( \hat{a}_m \hat{\rho}_{\bar{n}_0}^{\otimes M}   \hat{a}_m^\dagger  \right) \hat{\sigma}_0^{-1}\right\} \\&
= 
\int_0^1\dif s\left\{ \hat{a}_m \hat{\rho}_{\bar{n}_0}^{\otimes M}   \hat{a}_m^\dagger  \hat{\sigma}_0^{-2}\right\}\\
&=
\sum_{k_1,...,k_M=0}^\infty \frac{\bar{n}_0}{\bar{n}_0 + 1}(k_m+1)\\&
\quad\quad\bigtimes\int_0^1\dif s\prod_{i=1}^Mt_{k_i} \left[s\prod_{i=1}^Mt_{k_i} + (1-s)\right]^{-2}
\bigotimes_{i=1}^M\ket{k_i}\bra{k_i} \\
&=
\sum_{k_1,...,k_M=0}^\infty \frac{\bar{n}_0}{\bar{n}_0 + 1}(k_m+1)\bigotimes_{i=1}^M\ket{k_i}\bra{k_i} \\
\end{split}
\end{equation}
using the fact that $\int_0^1\dif sA\left[ sA+(1-s)\right]^{-2} = 1$ for $|A|<1$. Thus, the trace of $\hat{\rho}_{\bar{n}_0}^{\otimes M}\int_0^1\dif s\left\{ \hat{\sigma}_0^{-1}\left( \hat{a}_m \hat{\rho}_{\bar{n}_0}^{\otimes M}   \hat{a}_m^\dagger  \right) \hat{\sigma}_0^{-1}\right\}$ is equal to 

\begin{equation}
\begin{split}
&\sum_{k_1,...,k_M=0}^\infty \frac{\bar{n}_0}{\bar{n}_0 + 1}\prod_{i=1}^Mt_{k_i}(k_m+1) \\
&= 
\sum_{k_m=0}^\infty \frac{\bar{n}_0}{\bar{n}_0 + 1}t_{k_m}(k_m+1) \\
&=\bar{n}_0.
\end{split}
\end{equation}
Therefore, the trace of the first term of \eqref{eq:R_2} is $2M\tau\bar{n}_0^{-1}c_1x$. 
For the second term of \eqref{eq:R_2}, the integral over $s$ evaluates to identity such that the trace over the second term is simply the coefficient, $-2M\tau\bar{n}_0^{-1}c_1$.
he third term is evaluated similarly to the first term. The only non-zero contributions to the trace are a result of when $m_1=m_2$ such that the trace of the third term is $2M\tau\bar{n}_0^{-1}c_2^2$. 

The sum of all three traces is equal to the second derivative of $\hat{R}$ with respect to $b$ evaluated at $b=0$,  
\begin{equation}
\begin{split}
\frac{\dif ^2\hat{R}}{\dif b^2}&\bigg\vert _{b=0}=2M\tau\bar{n}_0^{-1}\left[ c_1x-c_1+c_2^2  \right] \\
&= 2M\tau\bar{n}_0^{-1}\left[ (1-\eta_2)x-(1-\eta_2)+(1-\eta_2)(1-x)  \right] \\
&= 0,
\end{split}
\end{equation}
therefore, the second term of the Taylor expansion of $\hat{R}$ is also zero.
\subsection{Third Term}
The third derivative of $\hat{R}$ with respect to $b$ is found using Lemma \ref{lemma:dlogint} and Lemma \ref{lemma:dinv} again,
\begin{align} \label{eq:R_3}
&\nonumber\frac{\mathrm{d}^3\hat{R}}{\mathrm{d}b^3}=
3\frac{\dif ^2}{\dif b^2}\hat{\rho}_\tau^{(M)} \int_0^1\mathrm{d}s\hat{\sigma}_1^{-1}(s)\frac{\dif }{\dif b}\hat{\rho}_\tau^{(M)}\hat{\sigma}_1^{-1}(s)\\
\nonumber&\phantom{=}-6\frac{\dif }{\dif b}\hat{\rho}_\tau^{(M)}\int_0^1s\mathrm{d}s\hat{\sigma}_1^{-1}(s)\frac{\dif }{\dif b}\hat{\rho}_\tau^{(M)}\hat{\sigma}_1^{-1}(s)\frac{\dif }{\dif b}\hat{\rho}_\tau^{(M)}\hat{\sigma}_1^{-1}(s)\\
\nonumber&\phantom{=}+3\frac{\dif }{\dif b}\hat{\rho}_\tau^{(M)}\int_0^1\mathrm{d}s\hat{\sigma}_1^{-1}(s)\frac{\dif ^2}{\dif b^2}\hat{\rho}_\tau^{(M)}\hat{\sigma}_1^{-1}(s)\\
\nonumber&\phantom{=}-3\hat{\rho}_\tau^{(M)}\int_0^1s\mathrm{d}s\hat{\sigma}_1^{-1}(s)\frac{\dif ^2}{\dif b^2}\hat{\rho}_\tau^{(M)}\hat{\sigma}_1^{-1}(s)\frac{\dif }{\dif b}\hat{\rho}_\tau^{(M)}\hat{\sigma}_1^{-1}(s)\\
\nonumber&\phantom{=}+6\hat{\rho}_\tau^{(M)}\int_0^1s^2\mathrm{d}s\hat{\sigma}_1^{-1}(s)\frac{\dif }{\dif b}\hat{\rho}_\tau^{(M)}\hat{\sigma}_1^{-1}(s)\\
\nonumber&\phantom{=}\quad\quad\quad\quad\quad\quad\quad\quad\quad\bigtimes\frac{\dif }{\dif b}\hat{\rho}_\tau^{(M)}\hat{\sigma}_1^{-1}(s)\frac{\dif }{\dif b}\hat{\rho}_\tau^{(M)}\hat{\sigma}_1^{-1}(s)\\
\nonumber&\phantom{=}-3\hat{\rho}_\tau^{(M)}\int_0^1s\mathrm{d}s\hat{\sigma}_1^{-1}(s)\frac{\dif }{\dif b}\hat{\rho}_\tau^{(M)}\hat{\sigma}_1^{-1}(s)\frac{\dif ^2}{\dif b^2}\hat{\rho}_\tau^{(M)}\hat{\sigma}_1^{-1}(s)\\
&\phantom{=}+\hat{\rho}_\tau^{(M)}\int_0^1\mathrm{d}s\hat{\sigma}_1^{-1}(s)\frac{\dif^3}{\dif{b}^3}\hat{\rho}_\tau^{(M)}\hat{\sigma}_1^{-1}(s) \\
\nonumber&\phantom{=}+ \frac{\dif^3}{\dif{b}^3}\hat{\rho}_\tau^{(M)}\ln\hat{\rho}_\tau^{(M)}-\frac{\dif^3}{\dif{b}^3}\hat{\rho}_\tau^{(M)}\ln\hat{\rho}_{\bar{n}_0}^{\otimes M}.
\end{align}

Note that all terms in \eqref{eq:R_3} contain factors of $\frac{\dif }{\dif b}\hat{\rho}_\tau^{(M)}$ or $\frac{\dif^3}{\dif{b}^3}\hat{\rho}_\tau^{(M)}$ such that $\frac{\mathrm{d}^3\hat{R}}{\mathrm{d}u^3}\Big\vert_{b=0}=0$ due to $\frac{\dif}{\dif{b}}\hat{\rho}_\tau^{(M)} \Big\vert _{b=0}=\frac{\dif^3}{\dif{b}^3}\hat{\rho}_\tau^{(M)} \Big\vert _{b=0}=0$.

\subsection{Fourth Term}
The fourth derivative of $\hat{R}$ with respect to $b$ is found using Lemma \ref{lemma:dlogint} and Lemma \ref{lemma:dinv} again, where the  following omits all terms with factors of $\frac{\dif}{\dif{b}}\hat{\rho}_\tau^{(M)} $ and $\frac{\dif^3}{\dif{b}^3}\hat{\rho}_\tau^{(M)}$,
\begin{align} 
\nonumber\frac{\mathrm{d}^4\hat{R}}{\mathrm{d} b^4}&=6\frac{\dif^2}{\dif{b}^2}\hat{\rho}_\tau^{(M)}\int_0^1\mathrm{d}s\hat{\sigma}_1^{-1}(s)\frac{\dif^2}{\dif{b}^2}\hat{\rho}_\tau^{(M)}\hat{\sigma}_1^{-1}(s)\\
\nonumber&\phantom{=}-6\hat{\rho}_\tau^{(M)}\int_0^1s\mathrm{d}s\hat{\sigma}_1^{-1}(s)\frac{\dif^2}{\dif{b}^2}\hat{\rho}_\tau^{(M)}\hat{\sigma}_1^{-1}(s)\\
\nonumber&\phantom{=}\quad\quad\quad\quad\quad\quad\quad\quad\bigtimes\frac{\dif^2}{\dif{b}^2}\hat{\rho}_\tau^{(M)}\hat{\sigma}_1^{-1}(s)\nonumber\\
&\phantom{=}+\hat{\rho}_\tau^{(M)}\int_0^1\mathrm{d}s\hat{\sigma}_1^{-1}(s)\frac{\dif^4}{\dif b^4}\rho_\tau^{(M)}\hat{\sigma}_1^{-1}(s)\\
\nonumber&\phantom{=}+\frac{\dif^4}{\dif b^4}\rho_\tau^{(M)}\ln\hat{\rho}_\tau^{(M)}
-\frac{\dif^4}{\dif b^4}\rho_\tau^{(M)}\ln\hat{\rho}_{\bar{n}_0}^{\otimes M}.
\end{align}
Setting $b=0$ yields,
\begin{equation} \label{eq:R_4}
\begin{split}
\frac{\dif^4 \hat{R}}{\dif{b}^4}\bigg\vert _{b=0} &= 
6\frac{\dif^2}{\dif{b}^2}\hat{\rho}_\tau^{(M)}\bigg\vert _{b=0}\int_0^1\dif s \;\hat{\sigma}_0^{-1}\left( \frac{\dif{}^2}{\dif{b}^2}\hat{\rho}_\tau^{(M)}\bigg\vert _{b=0} \right) \hat{\sigma}_0^{-1} \\
&\;\;\;\;
-6 
\hat{\rho}_{\bar{n}_0}^{\otimes M}\int_0^1s\dif s\; \hat{\sigma}_0^{-1}\left( \frac{\dif{}^2}{\dif{b}^2}\hat{\rho}_\tau^{(M)}\bigg\vert _{b=0} \right) \hat{\sigma}_0^{-1}\\&
\;\;\;\;\;\;\;\;\;\;\;\;\;\;\;\;\;\;\;\;\;\;\;\;\;\;\;\;\;\;\;\;\;\;\;\;\times
\left( \frac{\dif{}^2}{\dif{b}^2}\hat{\rho}_\tau^{(M)}\bigg\vert _{b=0} \right) \hat{\sigma}_0^{-1}\\
&\;\;\;\;
+
\hat{\rho}_{\bar{n}_0}^{\otimes M}\int_0^1ds\;  \hat{\sigma}_0^{-1}\left( \frac{\dif{}^4}{\dif{b}^4}\hat{\rho}_\tau^{(M)}\bigg\vert _{b=0}\right) \hat{\sigma}_0^{-1}  \\
\end{split}
\end{equation}  

The trace of the first term of \eqref{eq:R_4} is found using \eqref{eq:rho_tau_2}, 
\begin{equation} \label{eq:R_4_1st}
\begin{split}
&\frac{24\tau^2(c_1^2x^2 + 2c_1c_2^2x)}{\bar{n}_0^2}\left[ M^2  - M + M\frac{1+2\bar{n}_0}{1+\bar{n}_0} \right] \\&- \frac{48\tau^2c_1^2xM^2}{\bar{n}_0^2} 
+\frac{24\tau^2c_1^2M^2}{\bar{n}_0^2} -\frac{48\tau^2c_1c_2^2M^2}{\bar{n}_0^2}\\& + \frac{24\tau^2c_2^4M^2(1+2\bar{n}_0)}{\bar{n}_0^2(1+\bar{n}_0)}
\end{split}
\end{equation}  

The trace of the second term of \eqref{eq:R_4} is found similarly,  
\begin{equation} \label{eq:R_4_2nd}
\begin{split}
&\frac{-12\tau^2(c_1^2x^2 + 2c_1c_2^2x)}{\bar{n}_0^2}\left[ M^2  - M + M\frac{1+2\bar{n}_0}{1+\bar{n}_0}  \right] \\&+ \frac{24\tau^2c_1^2xM^2}{\bar{n}_0^2} 
-\frac{12\tau^2c_1^2M^2}{\bar{n}_0^2} +\frac{24\tau^2c_1c_2^2M^2}{\bar{n}_0^2} \\&- \frac{12\tau^2c_2^4M^2(1+2\bar{n}_0)}{\bar{n}_0^2(1+\bar{n}_0)}.
\end{split}
\end{equation}

The trace of the final term of \eqref{eq:R_4} is found using \eqref{eq:rho_tau_4},
\begin{equation} \label{eq:R_4_3rd}
\begin{split}
&\frac{-12\tau^2c_1^2x^2M}{\bar{n}_0^2}+ \frac{12\tau^2c_1^2x^2M^2}{\bar{n}_0^2}-\frac{12\tau^2c_1c_2^2xM}{\bar{n}_0^2}\\&+ \frac{24\tau^2c_1c_2^2xM^2}{\bar{n}_0^2} 
+\frac{12\tau^2c_1^2xM}{\bar{n}_0^2} -\frac{24\tau^2c_1c_2^2M^2}{\bar{n}_0^2} \\&+ \frac{12\tau^2M^2(c_1^2+c_2^2)}{\bar{n}_0^2}.
\end{split}
\end{equation}

The sum of \eqref{eq:R_4_1st}, \eqref{eq:R_4_2nd}, and \eqref{eq:R_4_3rd} equals the full trace, $\text{Tr}\left\{\frac{\dif^4 \hat{R}}{\dif{b}^4}\Big\vert _{b=0} \right\}=\frac{12\tau^2c_1^2M}{\bar{n}_0(1+\bar{n}_0)} \left[ M(1-x)^2 + x(2-x)\right]$. This yields the fourth term in the Taylor series $\frac{1}{4!}\frac{\dif^4 \hat{R}}{\dif{b}^4} = \frac{\tau^2c_1^2M}{2\bar{n}_0(1+\bar{n}_0)} \left[ M(1-x)^2 + x(2-x)\right]$. Substituting $\bar{n}_0 = \eta_2\bar{n}_{B_2}$, $b=\sqrt{\bar{n}_s'}$, and  $c_1=1-\eta_2$, we obtain
\begin{equation} 
\begin{split}
D(\hat{\rho}_\tau^{(M)} ||\hat{\rho}_{\bar{n}_0}^{\otimes M})=  
&\left[ M(1-x)^2 + x(2-x)\right]\\
&\bigtimes\frac{M\tau^2(1-\eta_2)^2\bar{n}_T^2}{2\eta_2\bar{n}_{B_2}(1+\eta_2\bar{n}_{B_2})} + o(\bar{n}_T^2).
\end{split}
\end{equation}  
Discarding higher order terms, the QRE can be expressed in terms of Bob's input displacement and thermal power by substituting $x=\frac{\bar{n}_S'}{\bar{n}_S'+\bar{n}_\alpha/M}$, and $\bar{n}_T= \bar{n}_\alpha/M + \bar{n}_{S}'$,
\begin{equation} 
\begin{split}
D(\hat{\rho}_\tau^{(M)} ||\hat{\rho}_{\bar{n}_0}^{\otimes M})=  
&\left[ \bar{n}_\alpha^2+2\bar{n}_\alpha\bar{n}_S^\prime+M\bar{n}_S^{\prime2}\right]\\
&\bigtimes\frac{\tau^2(1-\eta_2)^2}{2\eta_2\bar{n}_{B_2}(1+\eta_2\bar{n}_{B_2})}.
\end{split}
\end{equation}  
Thus, using the covertness criterion discussed in Section \ref{sec:covertnesscriterion} and by the additivity of QRE, we obtain the following bound on $\tau$,
\begin{equation} \label{}
\begin{split}
\tau\leq  \frac{1}{(1-\eta_2)} \sqrt{\frac{2\eta_2\bar{n}_{B_2}(1+\eta_2\bar{n}_{B_2})}{\bar{n}_\alpha^2+2\bar{n}_\alpha\bar{n}_S^\prime+M\bar{n}_S^{\prime2}}}\sqrt{\frac{\delta_\text{QRE}}{n/M}}.
\end{split}
\end{equation}  

\begin{rk}
In the case of only displacement with no thermal input, i.e. $\bar{n}_S'=0$, we have 
\begin{equation} 
\begin{split}
\tau\leq  \frac{\sqrt{2\eta_2\bar{n}_{B_2}(1+\eta_2\bar{n}_{B_2})}}{(1-\eta_2)\bar{n}_\alpha} \sqrt{\frac{\delta_\text{QRE}}{n/M}}
\end{split}
\end{equation}  
while in the case of only thermal input with no initial displacement, i.e. $\bar{n}_\alpha=0$, the QRE evaluates to 
\begin{equation} 
\begin{split}
\tau\leq  \frac{\sqrt{2\eta_2\bar{n}_{B_2}(1+\eta_2\bar{n}_{B_2})}}{(1-\eta_2)\bar{n}_S'} \sqrt{\frac{\delta_\text{QRE}}{n}}.
\end{split}
\end{equation} 
\end{rk}

\section{Corellator-based Scheme Performance}\label{ap:ccr_performance}
\subsection{Communication Task}

Here, we characterize the output distribution of the proposed classical correlation receiver in the case of discarding or transmitting the signal, $H_0$ and $H_1$, respectively. While we assume that channel one and channel two are lossless, this analysis can be extended to the more general case. On each mode, the receiver reports a 1 if a photon is detected on the idler and signal SPD channels, $c_i$ and $c_s$. Otherwise, the receiver reports a zero. Therefore the output distribution under each hypothesis $H_m\text{ , }m\in\{0,1\}$, is Bernouli with the following probability of success, 
\begin{align}
    p_{\rm CC}^{m} = {\rm Pr}({\text{click on }c_i}){\rm Pr}({\text{click on }c_s}|{\text{click on }c_i},H_m). 
\end{align}
The idler mode follows a thermal-state photon-number distribution, therefore $ {\rm Pr}({\text{click on }c_i})=1-\frac{1}{1+\bar{n}_s}$. Under the null hypothesis, the state returned to Alice is thermal with mean photon number per mode, $(1-\eta)\bar{n}_B$. Note these detection outcomes are independent. Therefore, ${\rm Pr}({\text{click on }c_s}|{\text{click on }c_i},H_0) = 1-\frac{1}{1+(1-\eta)\bar{n}_B}$, such that
\begin{align}
    p^0_{\rm CC} = \left(1-\frac{1}{1+\bar{n}_s}\right)\left(1-\frac{1}{1+(1-\eta)\bar{n}_B}\right)\label{eq:psucc0}
\end{align}
The state transmitted by Alice is given by a photon-number diagonal state 
\begin{equation}
\begin{split}
    \hat{\rho}_{S,1} &= \frac{1}{{{\rm Pr}({\text{click on }c_i})}}\sum_{k=1}^\infty \frac{\bar{n}_s^k}{(1+\bar{n}_s)^{k+1}}\ket{k}\bra{k}\\&
    = \sum_{k=1}^\infty p_{1}\left[k\right]\ket{k}\bra{k},
\end{split}
\end{equation}
where $ p_{1}\left[k\right] =  \left(1-\frac{1}{1+\bar{n}_s}\right)^{-1}\frac{\bar{n}_s^k}{(1+\bar{n}_s)^{k+1}}$. 
Under $H_1$ the returned state is a then a mixture 
\begin{align}
    \hat{\rho}_{A,1} = \sum_{k=1}^{\infty} p_1[k]\hat{\rho}_{\eta,\bar{n}_B}(k),
\end{align}
with $\hat{\rho}_{\eta,\bar{n}_B}(k)=\text{Tr}_W\left[\mathcal{E}^{(\eta,\bar{n}_B)}_{B\to AW}\left( \ket{k}\bra{k}\right) \right]$ being the state given by transmission of photon-number state $\ket{k}$ over the lossy thermal-noise bosonic channel with transmittance $\eta$ and mean environment photon number per mode $\bar{n}_B$. We can express $\hat{\rho}_{\eta,\bar{n}_B}(k)$ by its anti-normally ordered characteristic function,
\begin{equation}
\begin{split}
    \chi_{A,k}(\zeta)&=\chi_A^{\ket{k}}(\sqrt{\eta}\zeta)\times \chi_A^{\rm th, \bar{n}_B}(\sqrt{1-\eta}\zeta)\\
    &=e^{-\eta|\zeta|^2}\mathcal{L}_k\left(|\zeta|^2\right)\times e^{-(1+\bar{n}_B)(1-\eta)|\zeta|^2}\\
    &=e^{-(1+(1-\eta)\bar{n}_B)|\zeta|^2}\mathcal{L}_k\left(|\zeta|^2\right),
\end{split}
\end{equation}
 for Laguerre polynomial $\mathcal{L}_m(\cdot)$ of order $m$. 
 The state can be expressed as \begin{align}
     \hat{\rho}_{\eta,\bar{n}_B}(k)=\int\frac{{\dif}^2\zeta}{\pi}\chi_{A,k}(\zeta) \, e^{\zeta a^\dagger} e^{-\zeta^* a}.
 \end{align} 
Using the fact that $\bra{m}e^{\zeta a^\dagger} e^{-\zeta^* a}\ket{m}=\mathcal{L}_m(|\zeta|^2)$, we have 
 \begin{equation}
 \begin{split}
     \bra{m}\hat{\rho}_{\eta,\bar{n}_B}(k)\ket{m}&= \int\frac{{\rm d}^2\zeta}{\pi}e^{-(1+(1-\eta)\bar{n}_B)|\zeta|^2} \\ &\;\;\;\;\;\;\;\;\;\;\;\;\;\;\;\;\;\;\;
\mathcal{L}_k\left(|\zeta|^2\right)\mathcal{L}_m\left(|\zeta|^2\right)
     \end{split}
 \end{equation}
 For $k = 1$, 
\begin{equation}
\begin{split}
     \bra{m}\hat{\rho}_{\eta,\bar{n}_B}(1)&\ket{m} = \frac{((1-\eta)\bar{n}_B)^m}{(1+(1-\eta)\bar{n}_B)^{m+1}} \\&+ \eta\frac{((1-\eta)\bar{n}_B)^{m-1}(m-(1-\eta)\bar{n}_B)}{(1+(1-\eta)\bar{n}_B)^{m+2}}\label{eq:fockrepsinglerail},
\end{split}
\end{equation}
and the probability of getting a click on $c_s$ given $k=1$ is  
\begin{equation}
\begin{split}
    {\rm Pr}(\text{click on }c_s|k=1,H_1) = 1 &- \frac{1}{(1+(1-\eta)\bar{n}_B)}
    \\&+\eta\frac{1}{(1+(1-\eta)\bar{n}_B)^{2}}.
\end{split}
\end{equation}
Since probability of getting a click can only increase with higher $k$, we may lower bound the probability of success by 
\begin{equation}
\begin{split}
    p^1_{CC} \geq \left(1-\frac{1}{1+\bar{n}_s}\right)
    \bigg(1 &- \frac{1}{(1+(1-\eta)\bar{n}_B)}\\&+\eta\frac{1}{(1+(1-\eta)\bar{n}_B)^{2}}\bigg) \label{lower bound CCR}
\end{split}
\end{equation}
Since $p^1_{\rm CC}\geq p^0_{\rm CC}$ for all values of $\eta,\bar{n}_B,\bar{n}_s$, and we seek to discriminate between hypotheses, the lower bound in \eqref{lower bound CCR} effectively lower bounds receiver performance.
Note that this is an excellent approximation in the regime of interest, where $\bar{n}_s\ll1$.

The receiver sums the number of successes over the $M$-mode symbol before choosing a hypothesis based on a threshold. 
Thus, the receiver outputs a binomial distribution under each hypothesis, with $M$ trials and success probability $p_{\rm CC}^m$. Choosing a decision threshold induces a binary asymmetric channel between Bob and Alice. For the throughput analysis, the threshold and prior probabilities are chosen such that the mutual information over the induced channel is maximized.

\begin{rk}
    The performance of the detector is severely limited in the regime of  $(1-\eta)\bar{n}_B\gg1$ or $\bar{n}_B \gg 1$, since the probability of success under both hypotheses goes to $1$. This issue can be resolved with the use of photon number resolving detectors which modifies the outputs as $\sum_{i=1}^M N_i N_s$, where $N_i$ and $N_s$ are the random variables corresponding to photon number resolving detector output for the idler and signal channels, respectively. We omit this detector in our simulations using MODTRAN data since, in that regime, its performance is approximately the same as the SPD-based detector, and PNR detectors are less practical to use.
\end{rk}
\subsection{Ranging Task}
Next, we show that the SPD corellator-baser receiver can reliably estimate the path length, $l$, when it is unknown at the time of transmission without sacrificing the performance of the communication task. 

We state the following lemma,
\begin{lem}\label{lem:ranging_succ}
    Let \(K \sim \mathrm{Bin}(N,q_1)\) and \(X \sim \mathrm{Bin}(N,q_0)\) for Binomial distribution $\mathrm{Bin}(\cdot,\cdot)$ with \(q_1 > q_0\). Further, let $f_N^{(i)}(\cdot)$ and $F_N^{(i)}(\cdot)$ be the probability mass and cumulative distribution functions for $\mathrm{Bin}(N,q_i)$. Define $S_N 
    \triangleq \sum_{k=1}^N f_N^{(1)}[k] \big(F_N^{(0)}[k-1]\big)^{N_R}$. Then, for $N_R\in \mathcal{O}(N)$, 
\begin{equation}
    \lim_{N\to\infty}S_N = 1
\end{equation}
\end{lem} We defer the proof to the end of the section.
In this scenario, we analyze the probability of success for each output of the receiver when the distance to the target is incorrect. In this case, the clicks on detector $c_i$ are independent of clicks on $c_s$, so the probability of success is given by 
\begin{equation}
\begin{split}
    \tilde{p}^0_{\rm CC} &= {\rm Pr}({\text{click on }c_i}){\rm Pr}({\text{click on }c_s}|{\text{click on }c_i}) \\
    &=
    {\rm Pr}({\text{click on }c_i}){\rm Pr}({\text{click on }c_s}) \\
    &=\left(1-\frac{1}{1+\bar{n}_S}\right)
    \bigg(1-\frac{p_o\tau_M}{1+\eta\bar{n}_S +(1-\eta)\bar{n}_B}\\
    &\quad\quad\quad\quad\quad\quad\quad\quad\quad\quad\quad\quad-\frac{1-\tau_M}{1 +(1-\eta)\bar{n}_B}\bigg)\\
    &\leq\left(1-\frac{1}{1+\bar{n}_S}\right)
    \left(1-\frac{1}{1+\eta\bar{n}_S +(1-\eta)\bar{n}_B}\right)\\
\end{split}
\end{equation}
where $\tau_M$ is the prior probability of transmitting an `on' pulse within each sparsely selected symbol of size $M$, and $p_o=\frac{M}{N_R + 2M}$ represents the proportion of the signal that is aligned with the sparse code.
Let 
\begin{equation}
\begin{split}
        p* &= 1 - \tau_M\left(\frac{1}{(1+(1-\eta)\bar{n}_B)}+\eta\frac{1}{(1+(1-\eta)\bar{n}_B)^{2}}\right) \\
        &\quad\quad-\frac{1-\tau_M}{1+(1-\eta)\bar{n}_B}
\end{split}
\end{equation}
be the lower bound on success probability when the modes are matched perfectly conditioned on a click on $c_i$. 
When the modes are matched, the total probability of success for the induced Bernoulli distribution is 
\begin{equation}
\begin{split}
    \tilde{p}^1_{\rm CC} &= {\rm Pr}({\text{click on }c_i}){\rm Pr}({\text{click on }c_s}|{\text{click on }c_i}) \\
    &\geq\left(1-\frac{1}{1+\bar{n}_S}\right)\\ &\bigtimes\left((1-p_o)\left(1-\frac{1}{1+\eta\bar{n}_S +(1-\eta)\bar{n}_B}\right)+p_o p^*\right)
\end{split}
\end{equation}
The sensor sums over the Bernoulli trials, inducing a binomial distribution with $N=\lfloor (2N_R+M)B(n)\tau\rfloor$ trials on average and probability of success $\tilde{p}^m_{\rm CC}$, with $m\in\{0,1\}$ indicating whether the distance to target is selected correctly. Further, let $F^{(m)}[\cdot], f^{(m)}[\cdot]$ be the cumulative distribution function (c.d.f.) and probability mass function (p.m.f.) for each possibility $m$. We assume that the true ranges lies within the number of possible range values $N_R$. Since $\tilde{p}^1_{\rm CC} > \tilde{p}^0_{\rm CC}$, Alice decides the range to the target based on the highest outcome $k$ out of all chosen range values. Thus, her probability of successfully determining the range bin is given by 
\begin{align}
    p_{\rm succ} \geq \sum_{k=1}^{N}f^{(1)}[k] (F^{(0)}[k-1])^{N_R},\label{eq:psucc_rangin}
\end{align}
where the inequality is due to the possibility of a `tie' between different ranges.
Note that, in the limit of $n\to \infty$, $N\in\mathcal{O}(\sqrt{n})$, with $N_R$ fixed. Thus, by Lemma \ref{lem:ranging_succ}, we have that the range can be reliably determined in the limit of channel uses $n$. 

\begin{IEEEproof}[Proof of Lemma \ref{lem:ranging_succ}]
\begin{equation}
\begin{split}
S_N 
&\triangleq \sum_{k=1}^N f^{(1)}[k] \big(F^{(0)}[k-1]\big)^{N_R}\\&=\mathbb{E}_{K}\big[\big(F^{(0)}[K-1]\big)^{N_R}\big].
\end{split}
\end{equation}

Fix any \(\varsigma \in (0, q_1 - q_0)\) and define the event
\begin{align}
A_N &\triangleq \Big\{ \frac{K}{N} > q_0 + \varsigma \Big\}.
\end{align}

By Hoeffding's inequality, with \(\delta := q_1 - (q_0 + \varsigma) > 0\), we have
\begin{equation}
\begin{split}
\Pr(A_N^c) 
&= \Pr\Big( \frac{K}{N} \le q_0 + \varsigma \Big) \\
&\le \exp(-2\delta^2 N).
\end{split}
\end{equation}

For any \(k\) satisfying \(k/N > q_0 + \varsigma\), the Chernoff bound \cite[Ch.~12.9]{cover02IT} gives
\begin{equation}
\begin{split}
\Pr(X \ge k) 
&= \Pr\Big( \frac{X}{N} \ge \frac{k}{N} \Big) \\
&\le \exp\Big( - N D\Big(\frac{k}{N} \Big\| q_0 \Big) \Big) \\
&\le \exp(-c N),
\end{split}
\end{equation}
where $D(p \| q) = p \log\!\left(\frac{p}{q}\right)
  + (1-p)\log\!\left(\frac{1-p}{1-q}\right)$ is the binary relative entropy and \(c>0\) depends on \(\varsigma, q_0\).

Hence, for such \(k\),
\begin{equation}
    \begin{split}
F^{(0)}[k-1] 
&= \Pr(X \le k-1) \\
&= 1 - \Pr(X \ge k) \\
&\ge 1 - e^{-cN}.
\end{split}
\end{equation}

Since \((1 - a)^{N_R} \ge 1 - N_R a\) for \(a \in [0,1]\), we have
\begin{align}
\big(F^{(0)}[k-1]\big)^{N_R} \ge 1 - N_R e^{-c N}.
\end{align}

We may split the expectation in \(S_N\) according to \(A_N\):
\begin{equation}
\begin{split}
S_N 
&= \mathbb{E}\big[\big(F^{(0)}[K-1]\big)^{N_R} \mathbf{1}_{A_N}\big] \\
&\phantom{=}+ \mathbb{E}\big[\big(F^{(0)}[K-1]\big)^{N_R} \mathbf{1}_{A_N^c}\big] \\
&\ge (1 - N_R e^{-cN}) \Pr(A_N) + 0 \\
&\ge (1 - N_R e^{-cN}) (1 - e^{-2\delta^2 N}).
\end{split}
\end{equation}
Since both error terms vanish exponentially as \(N \to \infty\), we conclude
\begin{align}
\lim_{N \to \infty} S_N = 1.
\end{align}
\end{IEEEproof}
\end{document}